\RequirePackage{silence}
\documentclass[aps,prx,10pt,groupedaddress,twocolumn]{revtex4-2}
\usepackage[T1]{fontenc} 
\usepackage[english]{babel} 
\usepackage[utf8]{inputenc}

\usepackage[table,usenames,dvipsnames]{xcolor}
\usepackage[colorlinks=true,
			hypertexnames=false
			]{hyperref}
\PassOptionsToPackage{linktocpage}{hyperref} 

\addto\extrasenglish{   }
\usepackage{subfig}
\usepackage[section]{placeins}
\usepackage{amsmath,amssymb,mathtools,amsthm,dsfont}
\usepackage{bm}
\usepackage{easybmat} 
\usepackage{paralist}
\usepackage{enumitem}
\usepackage{graphicx}
\usepackage{acro}
\makeatother
\usepackage{tikz}
\usepackage{tikz-qtree}
\usetikzlibrary{positioning}

\newcommand{\pc}[2]{\p{#1 \mid #2}}
\newcommand{\expectation}[1]{\mathbb E[#1]}
\newcommand{\p}[1]{\mathbb{P}\left[#1\right]}

\newcommand{\primed}{^\prime}
\newcommand{\btheta}{\bm \theta}
\newcommand{\partialderiv}[2]{\frac{\partial #1}{\partial #2}}

\newcommand{\pbar}{\bar{p}}
\newcommand{\Sb}{\bm S}
\newcommand{\Eb}{\bm E}

\newcommand{\kron}[2]{\delta_{#1,#2}} 

\newcommand{\exti}[1]{^{(#1)}}
\newcommand{\resti}[1]{_{#1}}
\newcommand{\restib}[1]{_{|#1}} 
\newcommand{\wi}[1]{_{-#1}}
\newcommand{\eei}[1]{\bm e\exti{#1}} 
 
\newcommand{\Ei}[1]{E\resti{#1}}

\newcommand{\Ebi}[1]{\bm E\restib{#1}}

\newcommand{\Emap}{g}
\newcommand{\Smap}{\mathcal{T}}
\newcommand{\Mmap}{\mathcal{M}}
\newcommand{\Synset}{\mathbb{F}_2^l}
\newcommand{\pei}{\theta_{\bm e}^i}
\newcommand{\peia}[2]{\theta^{#2}_{#1}}
\newcommand{\dtheta}{D_{\bm \theta}}
\newcommand{\errset}{\mathcal{E}^l_n}

\newcommand{\Pei}[1]{\p{\Ei{i} = #1}}
\newcommand{\PeiS}[1]{\p{\Ei{i} = #1, \Sb}}
\newcommand{\PeicS}[1]{\pc{\Ei{i} = #1}{\Sb}}

\newcommand{\Sst}{\bm S^\ast}

\newcommand{\qht}{\hat{q}}
\newcommand{\wqht}{_{-\qht}}
\newcommand{\eeq}{\eei{q}}

\newcommand{\Eqht}{E\resti{\qht}}

\newcommand{\eht}{\hat{e}}
\newcommand{\Ewqht}{\bm E\wi{\qht}}
\newcommand{\Est}{\bm E^{\ast}}
\newcommand{\Estq}{E^{\ast}\resti{q}}
\newcommand{\Estqht}{E^{\ast}\resti{\qht}}
\newcommand{\Modst}{$\Sst$-modification}
\newcommand{\Extst}{$\vec \eht^{(\qht)}$-extension}

\newcommand{\nest}{n_{\mathrm{est}}}
\newcommand{\nit}{n_{\mathrm{it}}}
\newcommand{\nconcat}{n_{\mathrm{concat}}}

\renewcommand{\vec}[1]{\boldsymbol{#1}}
\DeclareMathOperator{\wt}{wt}
\newcommand{\e}{\ensuremath\mathrm{e}}
\renewcommand{\i}{\ensuremath\mathrm{i}}

\DeclareMathOperator*{\argmax}{arg\max}

\newtheorem{theorem}{Theorem} 
\newtheorem{lemma}{Lemma} 
\newtheorem{proposition}{Proposition} 
\newtheorem{notation}{Notation} 
\newtheorem{definition}{Definition} 
\newtheorem{corollary}{Corollary} 

\DeclareAcronym{EM}{short=EM,long=expectation maximization}
\DeclareAcronym{HEM}{short=HEM,long=hard assignment expectation maximization}
\DeclareAcronym{BP}{short=BP,long=belief propagation}
\DeclareAcronym{MSE}{short=MSE,long=mean squared error}
\DeclareAcronym{CRB}{short=CRB,long=Cramér-Rao bound}
\DeclareAcronym{SMC}{short=SMC,long=sequential Monte Carlo}

\makeatletter 
 \hypersetup{pdftitle = {Optimal noise estimation from syndrome statistics of quantum codes},
       pdfauthor = {Thomas Wagner, Hermann Kampermann, Dagmar Bruss, Martin Kliesch},
       pdfsubject = {Quantum information theory},
       pdfkeywords = {
                      quantum, error, correction, QEC, code, QECC, 
                      stabilizer, decoding, 
                      estimation, noise, from, syndrome, statistics, measurements
                      identifiability, weight, 
                      conditional independence, conditionally independent, 
                      perfect, five qubit, five-qubit, concatenated, code, 
                      belief propagation, 
                      expectation maximization,
                      Cramer-Rao, bound,
                      Fisher information, 
                      mean squared, error,
                      Steane code, Pauli noise,
                      logical error, decomposable error model, 
                      maximum likelihood
                      } 
      }
\makeatother

\begin{document}
\title{Optimal noise estimation from syndrome statistics of quantum codes}

\author{Thomas Wagner}
\email[]{thomas.wagner@uni-duesseldorf.de}
\author{Hermann Kampermann}
\author{Dagmar Bru\ss}
\author{Martin Kliesch}
\affiliation{Heinrich Heine University Düsseldorf}
\begin{abstract}
Quantum error correction allows to actively correct errors occurring 
in a quantum computation when the noise is weak enough.
To make this error correction competitive information about the specific noise is required.  
Traditionally, this information is obtained by benchmarking the device before operation. 
We address the question of what can be learned from only the measurements done during decoding. 
Such estimation of noise models was proposed for surface codes, exploiting their special structure, and in the limit of low error rates also for other codes. 
However, so far it has been unclear under what general conditions noise models can be estimated from the syndrome measurements. 
In this work, we derive a general condition for identifiability of the error rates. 
For general stabilizer codes, we prove identifiability under the assumption that the rates are small enough.  
Without this assumption we prove a result for perfect codes.
Finally, we propose a practical estimation method with linear runtime for concatenated codes. 
We demonstrate that it outperforms other recently proposed methods and that the estimation is optimal in the sense that it reaches the \acl{CRB}. 
Our method paves the way for practical calibration of error corrected quantum devices during operation. 
\end{abstract}

\maketitle

\section{Introduction}
Quantum error correction is an essential ingredient in quantum computing schemes. 
 When employing active quantum error correction via stabilizer codes, the decoding generally requires information about the error rates of all qubits. These error rates can vary significantly between qubits \cite{Tannu_VariabilityAwareNISQ} and they might even vary in time.
 In contrast to traditional benchmarking before operation, a new approach is to estimate error rates online from the syndrome statistics of the code itself \cite{Fujiwara_InstantaneousChannelEstimationCSS,
Fowler_ScalableExtractionOfErrorModelsFromQEC,
Huo_2017,
Wootton_QiskitBenchmarking,
FlorjanczykBrun_InSituAdaptiveEncoding,
OBrien_AdaptiveWeightEstimator,
Combes_InSituCharOfDevice}. It should be stressed that the syndrome statistics is the only information that can be measured without destroying the encoded information.
As pointed out by \citet{Fowler_ScalableExtractionOfErrorModelsFromQEC}, this results in a noise model that is directly applicable for the decoder.
 Furthermore, it allows for the tracking of time-varying error rates \cite{Huo_2017, OBrien_AdaptiveWeightEstimator}.
 Experimentally, online optimization of control parameters in a 9-qubit superconducting quantum processor has been demonstrated in a Google experiment \cite{Kelly_ScalableInSituCalibrationDuringErrorDetection}.
 
 However, apart from the work of \citet{OBrien_AdaptiveWeightEstimator}, there has been very little theoretical investigation of the estimation problem, see \autoref{sec:SO-Solution} for a detailed discussion. 
 For example, it is not clear for what combinations of noise models and codes the unknown parameters are identifiable from the syndrome statistics.
 Evidently, some restrictions must apply since estimating completely general noise would require measurements which destroy the logical state.
 For some codes and noise models, including surface codes with independent Pauli noise on each qubit, the analytical method developed by \cite{OBrien_AdaptiveWeightEstimator} proves parameter identifiability.
 On the other hand, for many other important codes such as the 5-qubit code \cite{Laflamme_5QubitCode}, the Steane code \cite{Steane_SteaneCode}, and more general color codes \cite{Bombin_ColorCodes} this method is not applicable.
 
 In this work, we address this question by deriving a general condition for parameter identifiability, and using it to explicitly prove results for the 5-qubit code and the Steane code.
 Furthermore, we introduce an explicit error rates estimator, similar to techniques employed in classical distributed source coding \cite{Zia_LDPCErrorRateEstimation}, for concatenated codes and simulate it on the concatenated 5-qubit code.
 This estimator outperforms previously proposed methods \cite{Wootton_QiskitBenchmarking,Huo_2017,Fowler_ScalableExtractionOfErrorModelsFromQEC} in this setting, because it does not require the assumption of very low error rates.

\subsection*{Stabilizer Codes \label{sec:StabilizerCodes}}
Let us introduce our notation while briefly summarizing stabilizer codes. The Pauli group $\mathcal{P}_n$ on $n$ qubits is the group of \emph{Pauli strings} generated by the Pauli operators $\{X,Y,Z,I\}$ with phases, 
\begin{equation}
\mathcal{P}_n = \{\epsilon \bigotimes_{i = 1}^n e\resti{i} \mid \epsilon \in \{\pm 1, \pm \i\}, e\resti{i} \in \{I,X,Y,Z\}\} \, . 
\end{equation}
The Pauli group modulo phases
\begin{equation}
\mathsf P_n = \mathcal{P}_n / \{\pm 1, \pm \i\}
\end{equation}
is called the \emph{effective Pauli group}. We denote the $i$-th tensor factor of $\bm e \in \mathsf P_n$ as $e\resti{i}$. The Pauli operator acting as $e \in \mathsf P_1$ on qubit $i$ and as the identity elsewhere is denoted $\bm e\exti{i} \in \mathsf P_n$.
A stabilizer code encoding $k = n-l$ qubits is defined by a commutative subgroup $\mathcal S$ of $\mathcal{P}_n$ with generators $\bm g_1,\dots,\bm g_{l}$ \cite{NielsenChuang}.
 The code-space is the simultaneous $+1$-eigenspace of the generators.
 Phases are generally not important for quantum error correction, so we consider data errors as elements of the effective Pauli group.
 For an error $\bm e \in \mathsf P_n$, we define the \emph{syndrome} $\Sb(\bm e) \in \mathbb{F}_2^{l}$ entry-wise by
\begin{equation}
\Sb(\bm e)_i \coloneqq 
\begin{cases} 
0, & \text{if $\bm g_i$ and $\bm e$ commute in $\mathcal P_n$, } 
\\ 
1, & \text{if $\bm g_i$ and $\bm e$ anti-commute in $\mathcal P_n$.} 
\end{cases}
\end{equation}
To correct an error $\bm e \in \mathsf{P}_n$, a recovery $\bm r \in \mathsf{P}_n$ is applied based on the measured syndrome. Since errors that only differ by stabilizers act equivalently on the encoded information, the recovery is successful if the equivalence class $[\bm e \bm r]$ is trivial, i.e. $[\bm e \bm r] = [I] \in \mathsf P_n / \mathcal S$.

\section{Identifiability Conditions \label{sec:IdentifiabilityConditions}}
In this section, we derive a general condition for the identifiability of the error rates. 
Then we prove for perfect codes that this condition is always fulfilled whenever the error rates are sufficiently close. 
\subsection{General Conditions \label{sec:GeneralIdentifiabilityConditions}}
We consider	a stabilizer code with $n$ qubits and $l$ stabilizer generators. Let us first define identifiability. Given is a parameterized noise model, mapping a vector of error rates $\btheta$ to a vector $(\p{\bm E}_{\bm E \in \mathsf P_n})$ specifying the probability $\p{\bm E}$ for each error $\bm E \in \mathsf P_n$. This induces the map $\Mmap: \btheta \mapsto (\p{\Sb})_{\Sb \in \Synset}$, mapping a parameter vector to the corresponding syndrome statistics via
\begin{equation} \label{eq:P[S]}
\p{\bm S} = \sum_{\bm E \in \mathsf P_n \, : \, \bm S(\bm E) = \bm S} \p{\bm E} \, ,
\end{equation}
where $\p{\bm S}$ is the induced probability of observing the syndrome $\bm S \in \Synset$. Error rates are identifiable from the syndrome statistics if the map $\Mmap$ is injective. 
This will usually not be the case, due to symmetry around error rates of $0.5$. 
However, we can still hope that the parameters are at least identifiable if we restrict to some region in the space of parameters $\btheta$. 
\begin{definition}[Local identifiability]
We say that error rates are \emph{locally identifiable} at $\bm \theta$ if there exists $\varepsilon > 0$ such that the map $\Mmap$ is injective on the ball ${B_\varepsilon(\bm \theta) \coloneqq \{\bm \theta\primed \mid \|\bm \theta\primed - \bm \theta \|_2 < \varepsilon \}}$.
\end{definition} 
For ease of exposition, we will focus in this section on independent single qubit Pauli noise, which is a simple but widely studied error model. A substantial generalization of \autoref{lemma:LowRatesIdent} and \autoref{lemma:GeneralIdentifiability} to much more general error models, including measurement errors, can be found in \autoref{sec:GeneralizedIdentResults}. For now let us assume that errors on the $i$-th qubit of the code are modeled by the Pauli channel
\begin{equation}
\rho \mapsto (1-\theta^i_X - \theta^i_Y - \theta^i_Z)\rho + \theta^i_X X \rho X + \theta^i_Y Y \rho Y + \theta^i_Z Z \rho Z \; .
\end{equation}
with $\theta^i_X,\theta^i_Z,\theta^i_Y \in [0,1]$ such that $\theta^i_X + \theta^i_Z + \theta^i_Y \leq 1$.
The parameter vector $\btheta$ for this error model is given by the error rates $(\theta^i_e)_{i \in \{1,\dots,n\},e\in\{X,Y,Z\}}$ of all non-trivial single qubit errors. 
By the inverse function theorem, local identifiability at $\btheta$ holds if and only if the Jacobian matrix  $J = D_{\bm \theta} \Mmap$ at $\btheta$ has full (column) rank. We will label the rows of the Jacobian by syndromes $\bm S$ and the columns by parameters $\theta^i_{e}$, and denote entries with square brackets, e.g. as $J[\bm S, \theta^i_{e}]$.
In the limit of low error rates, it is intuitive that identification of error rates is possible since a syndrome always arises from the matching single qubit error, and no combined errors occur.
 Thus, the only requirement is that the single errors can be identified from the syndromes. This just means that the code has distance at least 3, i.e. only trivial codes are excluded. This leads to the estimators proposed in Refs.~\cite{Wootton_QiskitBenchmarking,Fowler_ScalableExtractionOfErrorModelsFromQEC,Huo_2017}. We confirm this intuition by calculating the Jacobian of $\Mmap{}$ and checking its rank:
\begin{proposition}[Identifiability for small error rates]
\label{lemma:LowRatesIdent}
 For a quantum code subject to independent single qubit Pauli noise, error rates are locally identifiable at $\bm \theta = \bm 0$ if and only if $\bm S (\bm e) \neq \bm S(\bm e \primed)$ for every choice of two different single qubit errors $\bm e, \bm e\primed$.
\end{proposition}

A proof is provided in \autoref{sec:ProofLowRatesIdent}. Our first central result is an identifiability condition without the assumption of low rates.
This establishes a connection between local identifiability and the posterior distribution of errors for each qubit.

\begin{theorem}[General identifiability condition]
\label{lemma:GeneralIdentifiability}
Consider a quantum code subject to independent single qubit Pauli noise.
Assume that all error rates are non-zero and that $\p{\bm S} > 0$ for all syndromes $\bm S \in \Synset$. 
Then error rates are locally identifiable at $\bm \theta$ if and only if the matrix $\tilde{J}$ with entries
\begin{equation}
\tilde{J}[\bm S, \theta^i_{e}] = \frac{\p{E_i = e | \bm S}}{\p{E_i = e}} - \frac{\p{E_i = I | \bm S}}{\p{E_i = I}}
\label{eq:JTilde}
\end{equation}
has full column rank.
Here, $\p{E_i = e | \bm S}$ is the conditional probability that the $i$-th qubit is affected by the error $e \in \mathsf P_1$ given that the observed syndrome is $\p{\Sb}$.
\end{theorem}
The proof is provided in \autoref{sec:ProofGeneralIdentifiability}.

\subsection{Identifiability for Perfect Codes \label{sec:PerfectCodeIdent}}
We demonstrate the analytical application of \autoref{lemma:GeneralIdentifiability} by considering the class of \emph{perfect codes}.

\begin{definition}[Perfect single error correcting quantum code \cite{Gaitan2008_book}] 
A quantum code $C$ on $n$ qubits is called a \emph{perfect single error correcting} code if there is a bijection between the set of non-trivial single qubit errors and the set of non-trivial syndromes, i.e.\ there exists a bijective map
\begin{equation}
f: \{\eei{i} \mid e \in \{X,Y,Z\}, \, i \in \{1,\dots,n\}\} \rightarrow \Synset \setminus \{\bm 0\}
\end{equation}
\end{definition}
These codes are called ``perfect'' because they saturate the (quantum) Hamming bound.
A well known example of such a code is the 5-qubit code \cite{Laflamme_5QubitCode}. Other families of perfect codes are cyclic Hamming codes \cite{Gottesman_CyclicHammingCodes} and a class of twisted codes \cite{Bierbrauer_QuantumTwistedCodes}. The main result of this section is that error rates for such codes are locally identifiable around the points of equal rates, even for high error rates. This provides another concrete class of codes where identification of error rates is possible.
\begin{theorem}[Identifiability for perfect codes]
\label{lemma:PerfectCodeIdent}
Let $C$ be a perfect single error correcting quantum code on $n$ qubits subject to independent single qubit Pauli noise.
Then the error rates are locally identifiable around any point $\bm\theta$ with equal error rates, i.e.\ if there exists $p \in (0,1)$ such that $\theta^i_e = p$ for all $i$ and all $e \in \{X,Y,Z\}$. 
\end{theorem}
Note that the condition on $\bm \theta$ above does \emph{not} mean that we restrict ourselves to a simple single parameter model.
We still allow all estimated error rates to vary individually, but require that the actual error rates are close to being equal.
In order to prove \autoref{lemma:PerfectCodeIdent} via \autoref{lemma:GeneralIdentifiability}, we have to check the rank of the matrix $\tilde{J}$ given in \eqref{eq:JTilde}. Using Bayes Theorem, we can express its entries as
\begin{equation}
\tilde{J}[\Sb, \theta^i_e] = \frac{\pc{\Sb}{E\resti{i} = e} - \pc{\Sb}{E\resti{i} = I}}{\p{\Sb}} 
 \, .
\label{eq:JPerfectCode}
\end{equation}
The key insight, which might be of independent interest, is that most of the conditional probabilities in this expression are equal:
\begin{lemma}
\label{lemma:PerfectCodeEqualCondProbs}
Consider a perfect single error correcting code on $n$ qubits subject to independent single qubit noise where all error rates are equal. Let $e,e\primed \in \mathsf P_1$. Then for any syndrome $\Sb \in \Synset \setminus \{\bm 0\}$ and qubit $i$ such that $\Sb \neq \Sb(\eei{i})$ and $\Sb \neq \Sb((\bm e\primed)\exti{i})$ we have
\begin{equation}
\pc{\Sb}{E\resti{i}=e} = \pc{\Sb}{E\resti{i}=e\primed}
\end{equation}
\end{lemma}
The proof is provided in \autoref{sec:ProofPerfectCodeEqualCondProbs}.
This Lemma immediately implies that if $\Sb \neq \bm 0$ and $\Sb \neq \Sb(e_{i})$, then $\tilde{J}[\bm S, \theta^i_e] = 0$. We can ignore the case $\bm S = \bm 0$ due to normalization, and in the case $\Sb = \Sb(e_{i})$ we have $\tilde{J}[\bm S, \theta^i_e] \neq 0$. Thus the columns of $\tilde{J}$ are linearly independent unit vectors, i.e. $\tilde{J}$ has full rank. This proves \autoref{lemma:PerfectCodeIdent}.

The arguments of the proof straightforwardly generalize to other noise models as long as the perfect code condition is fulfilled, i.e.\ there is a bijection between syndromes and elementary errors. 
For example one could consider simple noise models where Pauli $X$ and Pauli $Z$ errors occur independently. The rates of such a model are locally identifiable on the Steane code around points of equal rates, since the Steane code reduces to the classical Hamming code when only one type of errors is considered. The Hamming code is known to be a perfect code.
Estimation of such a model on the Steane code was considered in Ref.~\cite{Huo_2017}. Thus, \autoref{lemma:PerfectCodeIdent} also provides a theoretical background for the results presented there.

\section{Numerical Estimation Method}
In this section, we complement the previous results with a practical estimation method, which is based on the combination of \ac{BP} and \ac{EM}. In the limit of low error rates, methods based on ``hard assignments'' were proposed independently by \cite{Wootton_QiskitBenchmarking,Fowler_ScalableExtractionOfErrorModelsFromQEC, Huo_2017}. They use either the recovery output by a (``hard'') decoder or the lowest weight error corresponding to a syndrome.
 Inspired by techniques from classical distributed source coding \cite{Zia_LDPCErrorRateEstimation}, we instead consider an estimation method that uses the full information about the distribution of errors given a certain syndrome, by combining a ``soft'' decoder \cite{Poulin_DecodingQuantumBlockCodes} with the \acl{EM} algorithm \cite{Dempster_ExpectationMaximization,Koller_PGM}.

\subsection{Concatenated Codes and Belief Propagation}
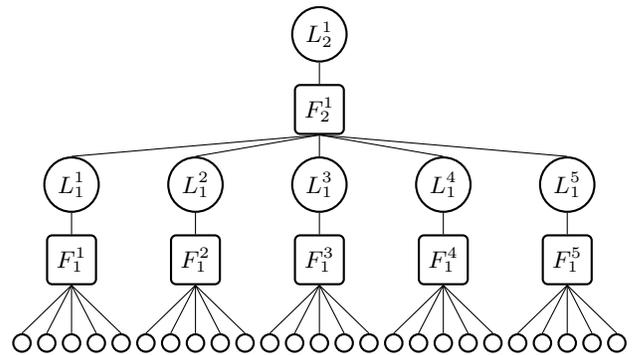
\begin{figure}[t]
\centering
\begin{tikzpicture}[rounded corners=.5ex,inner sep = .6ex] 
\tikzstyle{error}=[circle,thick,draw]
\tikzstyle{stabilizer}=[rectangle,thick,draw, minimum size = 2em]
\tikzstyle{errorfactor}=[circle,thick,draw, minimum size = 2em]
\tikzstyle{edge}=[very thick]

\tikzset{level distance=1cm}

\Tree[.\node[errorfactor]{$L^1_2$};
    [.\node[stabilizer]{$F^1_2$};
      [.\node[errorfactor]{$L^1_1$};
        [.\node[stabilizer]{$F^1_1$};
          [.\node[error]{};
          ]
          [.\node[error]{};
          ]
          [.\node[error]{};
          ]
          [.\node[error]{};
          ]
          [.\node[error]{};
          ]
        ]
      ]
      [.\node[errorfactor]{$L^2_1$};
        [.\node[stabilizer]{$F^2_1$};
          [.\node[error]{};
          ]
          [.\node[error]{};
          ]
          [.\node[error]{};
          ]
          [.\node[error]{};
          ]
          [.\node[error]{};
          ]
        ]
      ]
      [.\node[errorfactor]{$L^3_1$};
        [.\node[stabilizer]{$F^3_1$};
          [.\node[error]{};
          ]
          [.\node[error]{};
          ]
          [.\node[error]{};
          ]
          [.\node[error]{};
          ]
          [.\node[error]{};
          ]
        ]
      ]
      [.\node[errorfactor]{$L^4_1$};
        [.\node[stabilizer]{$F^4_1$};
          [.\node[error]{};
          ]
          [.\node[error]{};
          ]
          [.\node[error]{};
          ]
          [.\node[error]{};
          ]
          [.\node[error]{};
          ]
        ]
      ]
      [.\node[errorfactor]{$L^5_1$};
        [.\node[stabilizer]{$F^5_1$};
          [.\node[error]{};
          ]
          [.\node[error]{};
          ]
          [.\node[error]{};
          ]
          [.\node[error]{};
          ]
          [.\node[error]{};
          ]
        ]
      ]   
    ]
   ]
\end{tikzpicture}
\caption{The factor graph representation of the 2 times concatenated 5-qubit code. 
The circles depict variable nodes representing (logical) errors, i.e.\ each variable takes values in $\mathsf P_1$. 
The squares depict factor nodes representing the stabilizer checks.
}
\label{fig:FactorGraph5QubitCode}
\end{figure}
Let us briefly summarize concatenated codes and their maximum-likelihood decoding \cite{Poulin_DecodingQuantumBlockCodes}. We consider independent single qubit Pauli errors. 
 A concatenated quantum code is obtained by encoding each qubit of a quantum code again in the same code.
 This defines a tree structure, where the logical qubit of a code-block is a ``physical qubit'' in the next layer, as illustrated in \autoref{fig:FactorGraph5QubitCode} for the 5-qubit code.
 We can view this as a graphical representation of the probability distribution over all possible errors given the measured syndrome, called a \emph{factor graph}.
 The root node represents the total logical error.
 Maximum-likelihood decoding is done by computing its marginal distribution to find the most likely logical operator.
 Computation of marginal probabilities is efficiently possible using the \ac{BP} algorithm (see e.g. Ref. \cite{Bishop_ML}).
 \Ac{BP} works by passing messages along the edges of the graph. To compute the marginal of the root node it suffices to pass messages upwards, starting from the leaves.
 Doing an additional downwards pass, we can also calculate the marginals of the leaf nodes, i.e. the distribution of errors on a qubit given the measured syndrome.
 The  computational effort of this method scales linearly in the number of qubits.
 
\subsection{Error Rates Estimation and Expectation Maximization}

Starting from an initialization $\btheta^{(0)}$ of the estimated error rates and given a set $D$ of measured syndromes, we can calculate a new estimate of the error rates using the \ac{EM} algorithm, i.e. iterating the following steps until convergence.
\begin{enumerate}
\item Expectation step: Compute the \emph{expected sufficient statistics}
\begin{equation*}
M_E(e_i | \btheta^{(k)}) = \sum_{\Sb \in D} \p{E\resti{i} = e | \Sb, \btheta^{(k)}}
\end{equation*}
 based on the current estimate $\btheta^{(k)}$ of the error rates.
\item Maximization step: Compute a new estimate $\btheta^{(k+1)}$ of the error rates by normalizing the expected sufficient statistics:
\begin{equation}
(\theta^{(k+1)})^i_e = \frac{M_E[e_i | \btheta^{(k)}]}{\sum_{e\primed \in \mathsf P_1} M_E[e\primed_i | \btheta^{(k)}]}
\label{eq:EMMaximization}
\end{equation}
\end{enumerate}
Computationally, the main effort is in calculating the conditional probabilities needed for the expectation step. The key point is that this can be done efficiently using \ac{BP}. In an online estimation setting, the first iteration of \ac{EM} introduces almost no overhead, since the marginals calculated during decoding can be used. Further iterations require re-decoding of the syndromes and are thus roughly as expensive as decoding.
We will also compare our estimator with the ``hard assignment'' method \cite{Huo_2017,
Fowler_ScalableExtractionOfErrorModelsFromQEC,
Wootton_QiskitBenchmarking}, which is the best known scalable method. 
We extend this method slightly by allowing for multiple iterations. It can then be expressed as a variant of \ac{EM}, called \ac{HEM} (see Ref.~\cite{Koller_PGM}). 
 It consists of iterating the steps:
\begin{enumerate}
\item For each syndrome $\Sb \in D$, compute the most likely error
\begin{equation*}
{\Eb_{map}(\Sb) = \argmax_{\Eb \in \mathsf P_n}(\p{\Eb | \Sb, \btheta^{(k)}})}
\end{equation*} 
\item Obtain the new error rates by counting how often each single qubit error appears:
\begin{equation*}
(\theta^{(k+1)})^i_e = \frac{\sum_{\Sb \in D} \kron{E_{map}(\Sb)\resti{i}}{e} }{|D|} \; .
\end{equation*}
Here $\delta$ is the Kronecker-delta.
\end{enumerate}
 Instead of the marginals, only the most likely error for each syndrome is considered.
It can be computed efficiently using the max-sum algorithm which works similar to \ac{BP}, see e.g. \cite{Bishop_ML}.

\subsection{Numerical Results}
\label{sec:NumericalResults}
\begin{figure*}[tb]
\includegraphics[scale=1]{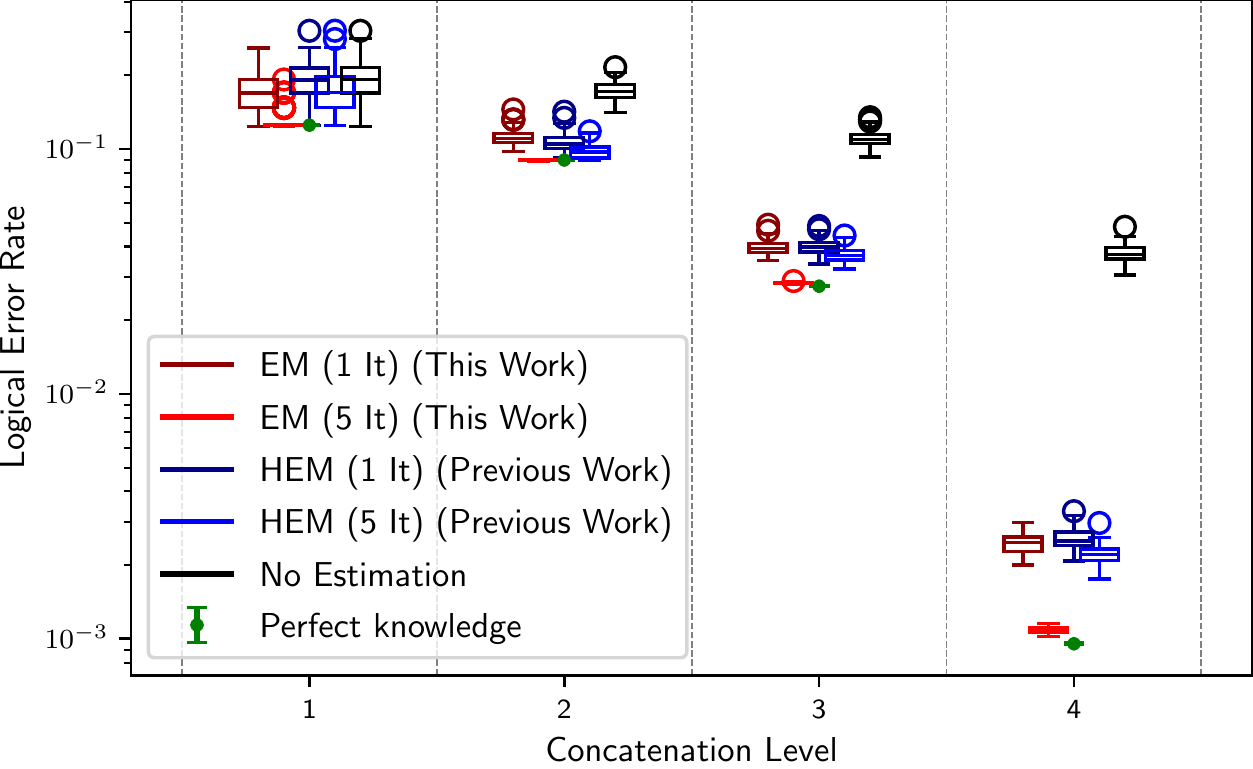}
\caption{Logical error rate of the maximum likelihood decoder. Each point is a box-plot including 100 runs with random initializations and estimation data, except for the perfect knowledge decoder, where the error bars indicate a $95 \%$ Clopper-Pearson confidence interval. 
The boxes extend from the lower to the upper quartile of values, with a line at the median. The whiskers extend to the last data point within 1.5 interquartile ranges of the box in each direction. Outliers beyond this are shown individually as circles. 
The parameters were $p  = 0.13$, $\alpha = 20$ and $\nest = 10^3$. \label{fig:ErrorRates_p013_A20}}
\end{figure*}

\begin{figure*}[t]
\includegraphics[scale=1]{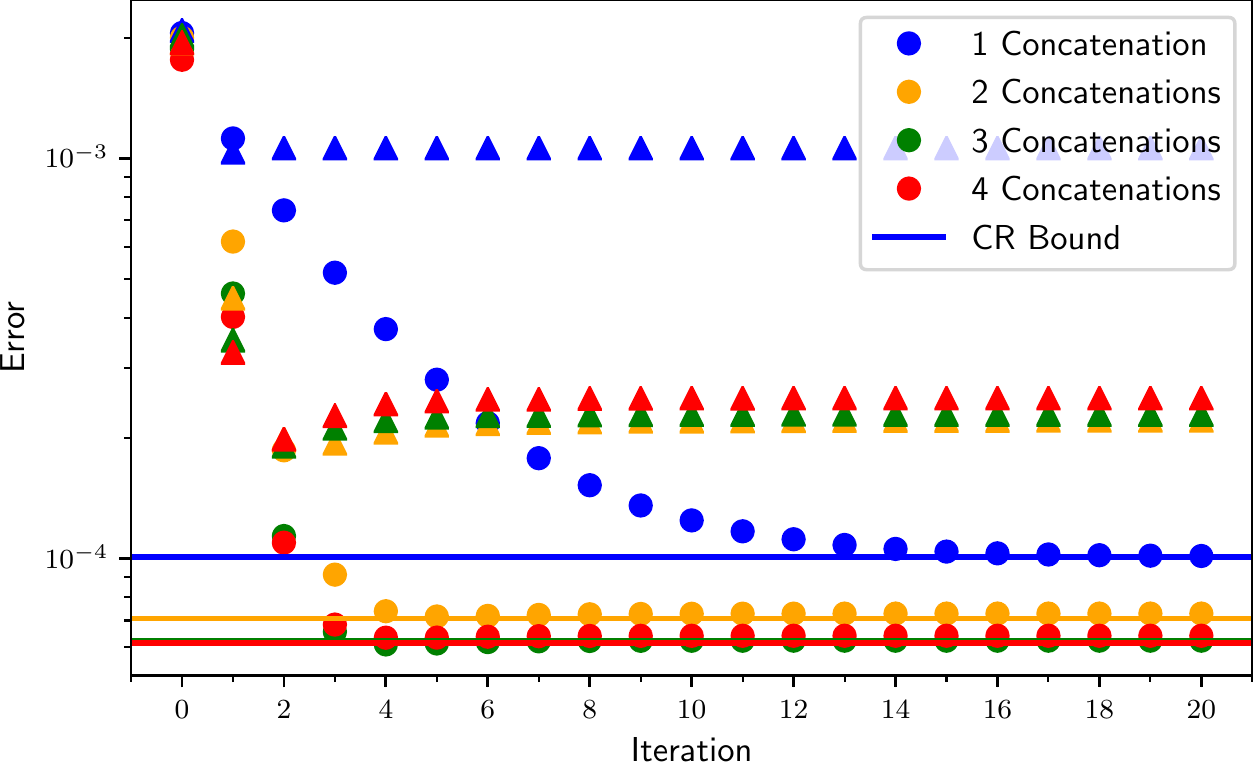}
\caption{Comparison of the \ac{MSE} in $\theta^1_X$ between \ac{EM} (circles, this work) and \ac{HEM} (triangles, previous work). The \acl{CRB} for each concatenation level is indicated by a line in the matching color.
The parameters were $p = 0.13$, $\alpha = 20$, and $\nest = 1000$. 
\label{fig:MSEOverConcat_A20_p013_nest1000}
}
\end{figure*}
In the following, we present a numerical comparison of our estimator (\ac{EM}) and the ``hard assignment'' estimator (\ac{HEM}).
 In light of our previous identifiability results, we consider the 5-qubit code, concatenated with itself, subject to independent depolarizing noise with error rate $p$ on each qubit.
 Extending the method to a phenomenological noise model with measurement errors is straightforward, and some results are shown in \autoref{sec:AdditionalNumericalResults}.
 We initialize the algorithm randomly around the actual rates, with a precision controlled by a real parameter $\alpha$ (higher is more accurate).
To be precise, for each qubit $i$, we sample error rates $\btheta^i$ from a Dirichlet distribution
\begin{equation}
\p{\btheta^i} = \frac{1}{B(\vec{\alpha})} \prod_{e \in \mathsf P_1} (\theta^i_e)^{\alpha_e}
\end{equation}
where $\alpha_I = (1- 3p)\alpha$, $\alpha_X = \alpha_Y = \alpha_Z = p\alpha$ and $B(\alpha)$ is a normalization constant.
Such an initialization could be obtained from previous benchmarking or an educated guess. We then run the estimator for $\nit$ iterations on a data set of $\nest$ syndromes generated from the actual distribution. Using a fixed initialization and random actual error rates was also tested for $\alpha = 20, \nest = 1000$ and did not significantly change the \ac{MSE} of the estimate of the parameter vector. We chose $p = 0.13$ which is close to the threshold of the code \cite{Poulin_DecodingQuantumBlockCodes, Rahn_ConcatCodes}, both because we are interested in the regime of high error rates, and because estimating logical error rates is difficult in the regime of low rates. A comparison of logical error rates before and after the estimation, using a relatively bad initialization, is shown in \autoref{fig:ErrorRates_p013_A20}. We also compare with the ``perfect knowledge decoder'' that is given knowledge of the actual error rates. Logical error rates were estimated by decoding $10^5 - 10^6$ random errors, except for the perfect knowledge decoder where $10^8$ random errors were used.
A clear improvement is observed even after 1 iteration, and for 5 iterations \ac{EM} was able to reach close to optimal error rates, while \ac{HEM} showed no further improvement after the first iteration.
We also confirmed that the \ac{MSE} of the \ac{EM} estimator is optimal in the sense that it reaches the \acl{CRB}, which lower bounds the \ac{MSE} of  any unbiased estimator (\autoref{fig:MSEOverConcat_A20_p013_nest1000}, more details and results in \autoref{sec:AdditionalNumericalResults}). The \ac{HEM} estimator showed significantly higher \ac{MSE}.
Finally, we note that since it is a form of maximum-likelihood estimation we expect the estimator to be robust in case of model-misspecification \cite{White_MLModelMismatch} -- quantifying the robustness is left for future research.  

\FloatBarrier
\clearpage
\onecolumngrid
\section{Details and Proofs}
In this section, we will give further details and generalizations on some topics and provide all the proofs that were previously omitted. Furthermore, we present more extensive numerical tests of our estimator.
\subsection{Analytical Solution Under a Conditional Independence Assumption}
\label{sec:SO-Solution}
\citet{OBrien_AdaptiveWeightEstimator} have derived an analytical solution of the estimation problem for certain models. 
Here, we re-derive this solution in a slightly more general setting and discuss the underlying assumptions and limitations by giving examples of quantum codes that cannot be treated in this way.

 The estimation method is considered for a circuit noise model, where errors can affect each part of the error correction circuit, including measurements.
\begin{definition}[Independent binary circuit noise]
Let $\{X_q\}_{q=1,\dots,m}$ denote a collection of (multi-qubit) Pauli errors, where each error may affect one or multiple sites in the error detection circuit. 
Under \emph{independent binary circuit noise}, each error occurs independently, and the error $X_q$ occurs with probability $\theta_q$.
\end{definition}
The errors in $\{X_q\}_{q=1,\dots,m}$ will also be referred to as \textit{elementary errors}.

In such a model, the errors can be treated as binary variables, where $X_q = 1$ with probability $\theta_q$ and $X_q = 0$ with probability $1-\theta_q$. Furthermore, the outcomes of the stabilizer generator measurements can be denoted by binary variables $S_i$, where $S_i = 1$ if the total error anti-commutes with the $i'$th generator and $S_i = 0$ otherwise.

Then, the rates of errors that affect multiple stabilizers can be estimated using the following proposition.

\begin{proposition}
\label{lemma:SO-1}
Consider a stabilizer code subject to independent binary circuit noise. Let $S_1$,$S_2$ be two syndrome bits and $X$ be an elementary error such that the following three conditions are fulfilled:
\begin{enumerate}
\item $\pc{S_1 = S_2}{X} = \p{S_1 = S_2}$
\item $\pc{S_i = 1}{X} = \pc{S_i = 0}{\bar X}$ for $i = 1, 2$
\item $S_1 \bot S_2 | X$, i.e.\ $S_1$ is conditionally independent of $S_2$ given $X$
\end{enumerate}
where $\bar X = 1 - X$.
Then
\begin{equation}
\p{X = 1}\p{X = 0} = \frac{\expectation{S_1S_2} - \expectation{S_1}\expectation{S_2}}{1 -  2\expectation{S_1 \oplus S_2}}
\label{eq:SO1}
\end{equation}
where $\oplus$ is addition modulo 2 and $\expectation{\, \cdot\,}$ denotes expectation values.
\end{proposition}
The idea is that the correlation between $S_1$ and $S_2$ gives us the rate of the error $X$. Note that the first two conditions are automatically fulfilled for any error $X$ that anti-commutes with both $S_1$ and $S_2$. The third condition however is interesting. It essentially states that $X$ is the only elementary error in our noise model that affects both $S_1$ and $S_2$.
\begin{proof}[Proof of \autoref{lemma:SO-1}]
Since the syndromes are binary variables, we have
\begin{flalign*}
 &\expectation{S_1S_2} - \expectation{S_1}\expectation{S_2}   \\
 = &\p{S_1 = 1, S_2 = 1} - \p{S_1 = 1}\p{S_2 = 1}  \\
\end{flalign*}
This can be rewritten using the law of total probability.
\begin{flalign*}
 = &\sum_{X = 0,1} \pc{S_1 = 1, S_2 = 1}{X}\p{X} \\
   &- \sum_{X,X\primed = 0,1}\pc{S_1 = 1}{X}\pc{S_2 = 1}{X\primed}\p{X}\p{X\primed} 
    \\
\end{flalign*}
Now we regroup the second term.
\begin{flalign*}
 = &\sum_{X = 0,1} \pc{S_1 = 1, S_2 = 1}{X}\p{X} \\
 &- \sum_{X,X\primed = 0,1 | X = X\primed}\pc{S_1 = 1}{X}\pc{S_2 = 1}{X\primed}\p{X}\p{X\primed}  \\
 & - \sum_{X,X\primed = 0,1 | X \neq X\primed}\pc{S_1 = 1}{X}\pc{S_2 = 1}{X\primed}\p{X}\p{X\primed} \\
  = &\sum_{X = 0,1} \pc{S_1 = 1, S_2 = 1}{X}\p{X}  \\
  &- \sum_{X = 0,1}\pc{S_1 = 1}{X}\pc{S_2 = 1}{X}\p{X}\p{X} \\
  &- \sum_{X= 0,1}\pc{S_1 = 1}{X}\pc{S_2 = 1}{\bar X}\p{X}\p{\bar X } \\
\end{flalign*}
Finally, we use assumptions 1,2 and 3 to finish the calculation.
\begin{flalign*}
 =& \p{X=1}\p{\bar X = 1}\left(\sum_{X = 0,1}\pc{S_1 = 1}{X}\pc{S_2 = 1}{X} - \sum_{X = 0,1}\pc{S_1 = 1}{X}\pc{S_2 = 1}{\bar X}\right) \text{\quad (by 3)} \\
 =& \p{X = 1}\p{\bar X = 1}\left(\pc{S_1 = S_2}{X = 1} - \pc{S_1 \neq S_2}{X = 1}\right) \text{\quad (by 2 for the second term)}\\
 =& \p{X = 1}\p{\bar X = 1}\left(1 - 2\p{S_1 \neq S_2}\right) \text{\quad (by 1)}  \\ 
 =& \p{X = 1}\p{\bar X = 1}\left(1 - 2\expectation{S_1 \oplus S_2}\right)
\end{flalign*}
where we also used $\p{\bar X} = 1 - \p{X}$.
This is equivalent to \eqref{eq:SO1}.
\end{proof}

Errors that only affect a single stabilizer can be estimated once the other rates are known, using the following proposition,
\begin{proposition}
\label{lemma:SO-2}
Let $S$ be a stabilizer and let $\{X_1,\dots,X_k\}$ be the set of all elementary errors in our noise model that anti-commute with $S$. Then,
\begin{equation}
\prod_{i = 1}^k (1 - 2\p{X_i = 1}) = (1 - 2\expectation{S})
\end{equation}
\end{proposition}
\begin{proof}
By assumption, the outcome of measuring $S$ is completely determined by the errors $X_1,\dots,X_k$. Therefore,
\begin{flalign*}
&(1 - 2\expectation{S}) \\
= &1 - 2\p{X_1 \oplus \dots \oplus X_k = 1} \\
\end{flalign*}
Since the elementary errors are independent we can factor out one of them.
\begin{flalign*}
= &1 - 2\left(\p{X_1 = 1}\p{\bigoplus_{i = 2}^k X_i = 0} + \p{X_1 = 0}\p{\bigoplus_{i = 2}^k X_i = 1}\right)\\
\end{flalign*}
Applying $\p{X_1 = 0} = 1 - \p{X_1 = 1}$ and doing some algebra leads to
\begin{flalign*}
= &1 - 2\left(\p{X_1 = 1}(1 - \p{\bigoplus_{i = 2}^k X_i = 1}) + (1 - \p{X_1 = 1})\p{\bigoplus_{i = 2}^k X_i = 1} \right) \\
=& 1 - 2\left(\p{X_1 = 1}(1 - 2\p{\bigoplus_{i = 2}^k X_i = 1}) + \p{\bigoplus_{i = 2}^k X_i = 1} \right) \\
=& 1 - 2\p{X_1 = 1} + 4\p{X_1 = 1}\p{\bigoplus_{i = 2}^k X_i = 1} -2\p{\bigoplus_{i = 2}^k X_i = 1} \\
=& \left(1 -2\p{X_1 = 1}\right) \left(1 - 2\p{\bigoplus_{i = 2}^k X_i = 1}\right)
\end{flalign*}
The claim now follows by induction.
\end{proof}
If e.g.\ the rates of $X_2,\ldots,X_k$ are already determined by using the estimation from the previous section, \autoref{lemma:SO-2} can be used to estimate $X_1$.

The estimation using propositions \ref{lemma:SO-1} and \ref{lemma:SO-2} is in closed form, however there are some limitations.
 First of all, the assumption of binary noise is relatively restrictive. For example, such a model does not include the commonly used depolarizing noise, since the probability of a Pauli $Y$ error is not the product of the probabilities of  $X$ and $Z$ errors.
 It is possible to work around this problem to some extent by modeling depolarizing noise as independent $X$,$Z$ and $Y$ errors with some effective rates, which works for low error rates.
 The second problem is that one only considers correlations between pairs of stabilizers, but not higher order correlations.
 This is generally not sufficient to fully characterize a code.
 For example, considering the well known 5 qubit code subject to independent Pauli noise on each qubit, there are 15 parameters to be estimated (the probabilities of each of the 3 non-trivial Pauli errors for each of the 5 qubits), while the two propositions provide at best $\binom{4}{2} + 4 = 10$ equations.
 However, we have shown in the main text that it is possible to estimate error rates of this code at least in certain parameter regimes (\autoref{lemma:PerfectCodeIdent}).
 Furthermore, \autoref{lemma:SO-1} requires that one can find pairs of stabilizers that are only correlated by a single elementary error.
 It is not always possible to find such pairs. As an example, consider the 7 qubit Steane code subject to only independent Pauli $X$ errors on each qubit.
 The stabilizers of this code are illustrated in \autoref{fig:7QubitCodeGraph}.
\begin{figure}[tb]
\newlength{\nodedist}
\setlength{\nodedist}{4cm}
\newlength{\triangleheight}
\setlength{\triangleheight}{3.46cm} 
\centering
\begin{tikzpicture}[node distance=\nodedist]
\tikzstyle{error}=[circle,thick,draw=black!100,fill=white!100,minimum size=4mm]
\tikzstyle{stabilizer}=[rectangle,thick,draw=black!100, fill=white!100,minimum size=6mm]
\tikzstyle{edge}=[very thick]

\node [stabilizer](s1){$S_1$};
\node [error,shift=({0:0.5\nodedist})s1](e1){$X_6$};
\node[stabilizer,shift=({0:\nodedist})e1](s2){$S_2$};
\node[error,shift=({60:0.5\nodedist})s1](e3){$X_5$};
\node[stabilizer,shift=({60:\nodedist})s1](s3){$S_3$};
\node[error] at ([shift=({120:0.5\nodedist})]s2)(e2){$X_3$};
\node[error] at ([shift=({90:0.5\triangleheight})]e1)(e4){$X_7$};
\node[error] at ([shift=({90:0.5\nodedist/2})]s3)(e5){$X_1$};
\node[error] at ([shift=({-45:0.5\nodedist/2})]s2)(e6){$X_2$};
\node[error] at ([shift=({225:0.5\nodedist/2})]s1)(e7){$X_4$};

\draw[edge](s1) to (e1);
\draw[edge](s1) to (e3);
\draw[edge](s1) to (e7);
\draw[edge](s1) to (e4);

\draw[edge](s2) to (e1);
\draw[edge](s2) to (e2);
\draw[edge](s2) to (e6);
\draw[edge](s2) to (e4);

\draw[edge](s3) to (e2);
\draw[edge](s3) to (e3);
\draw[edge](s3) to (e5);
\draw[edge](s3) to (e4);

\end{tikzpicture}
\caption{Errors $X_i$ and stabilizers $S_i$ for the 7 qubit Steane code with only $X$ errors (only the 3 relevant stabilizers are shown). A connection between an error and a stabilizer means that they anti-commute.}
\label{fig:7QubitCodeGraph}
\end{figure}
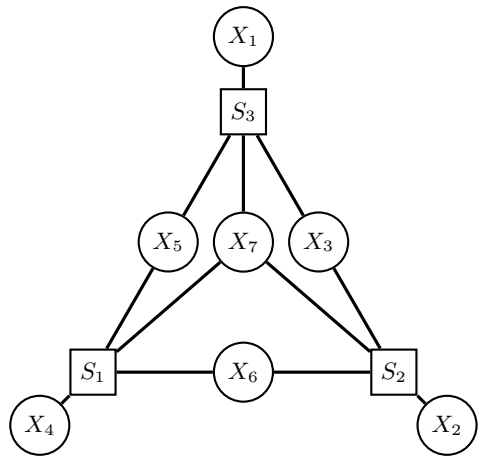
We see that because of the central error node $X_7$, there are no two stabilizers that are connected only through a single elementary error. Therefore we cannot apply \autoref{lemma:SO-1} here. However, \autoref{lemma:PerfectCodeIdent} implies that parameters of this model are identifiable at least in a certain regime. Note that similar problems occur for color codes, since the Steane code is the smallest example of a color code \cite{Bombin_ColorCodes}.

\subsection{Generalized Identifiability Results}
\label{sec:GeneralizedIdentResults}
In this section, we provide generalized versions of \autoref{lemma:LowRatesIdent} and \autoref{lemma:GeneralIdentifiability} as well as their proof. Furthermore, we provide the proof of \autoref{lemma:PerfectCodeIdent}.
\subsubsection{Formal Definition of Error Model}
\label{sec:FormalErrorModel}
We consider a quite general error model that includes independent single qubit Pauli noise as a special case. There are two main underlying assumptions. The first is that errors on the data qubits and syndrome bits are stochastic Pauli errors and bit-flips, which is common in the treatment of quantum error correction codes. The second is that there is some independence between different kinds of errors, which is both of fundamental importance for error correction and often physically reasonable. The first assumption implies that errors can be modeled as elements of the  group $\errset{} = \mathsf P_n \times  \Synset$, where the first component represents a Pauli error on the data and the second component represents a bit-flip on the measured syndrome. The product in this group is thus $(\bm e,\bm f), (\bm e\primed,\bm f\primed) \mapsto (\bm e \bm e\primed, \bm f \oplus \bm f\primed)$ and the identity element is $I = (I_{\mathsf P_n},\bm 0)$. The syndrome of $(\bm e,\bm f) \in \errset$ is $\Sb(\bm e) \oplus \bm f$.
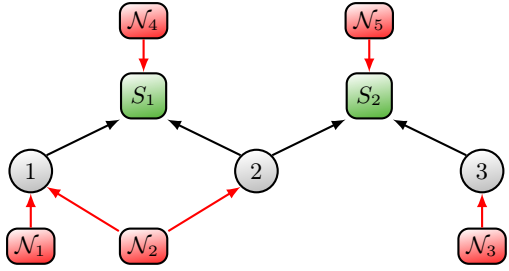
\begin{figure}[t]
\centering
\newlength{\horizontalnodedist}
\setlength{\horizontalnodedist}{3cm}
\newlength{\verticalnodedist}
\setlength{\verticalnodedist}{2cm}
\colorlet{mygreen}{OliveGreen!90!blue}%
\tikzset{
  sa/.style = {shading = axis,shading angle=10},
  blau/.style = {sa,top color=niceblue!12,%
                 bottom color=niceblue!88, sa},
  rot/.style = {top color=red!5,%
            bottom color=red!80,sa},
  gruen/.style = {top color=mygreen!5,%
            bottom color=mygreen!60,sa},
  hellgruen/.style = {top color=mygreen!1,%
            bottom color=mygreen!20},
  hellgrau/.style = {top color=gray!0,%
             bottom color=gray!50,sa},
  lila/.style = {top color=violet!3,%
             bottom color=violet!50,sa},
  gelb/.style = {top color=Dandelion!3,%
            bottom color=Dandelion!50,sa},
}
\begin{tikzpicture}[node distance =  \nodedist,on grid,
            rounded corners,
            ]
\tikzstyle{qubit}=[circle,thick,draw=black!100,hellgrau,minimum size=2mm,anchor=center]
\tikzstyle{stabilizer}=[rectangle,thick,draw=black!100, gruen,minimum size=6mm,anchor=center]
\tikzstyle{error}=[rectangle,thick,draw=black!100,rot,minimum size=2mm,anchor=center]
\tikzstyle{edge}=[thick]
\tikzstyle{arrow}=[->, >= latex, thick]
\tikzstyle{arrowr}=[arrow, red]
\node[qubit](q1){1};
\node[qubit](q2)[right = \horizontalnodedist of q1]{2};
\node[qubit](q3)[right = \horizontalnodedist of q2]{3};
\node[stabilizer](s1)[above right = 0.5\verticalnodedist and 0.5\horizontalnodedist of q1]{$S_1$};
\node[stabilizer](s2)[above right = 0.5\verticalnodedist and 0.5\horizontalnodedist of q2]{$S_2$};
\node[error](e1)[below = 0.5\verticalnodedist of q1]{$\mathcal{N}_1$};
\node[error](e3)[below = 0.5\verticalnodedist of q3]{$\mathcal{N}_3$};
\node[error](e12)[below right = 0.5\verticalnodedist and 0.5\horizontalnodedist of q1]{$\mathcal{N}_2$};
\node[error](es1)[above = 0.5\verticalnodedist of s1]{$\mathcal{N}_4$};
\node[error](es2)[above = 0.5\verticalnodedist of s2]{$\mathcal{N}_5$};

\draw[arrow](q1.north east) -- (s1.south west);
\draw[arrow](q2.north west) -- (s1.south east);
\draw[arrow](q2.north east) -- (s2.south west);
\draw[arrow](q3.north west) -- (s2.south east);

\draw[arrowr](e1.north) -- (q1.south);
\draw[arrowr](e3.north) -- (q3.south);
\draw[arrowr](e12) -- (q1.south east);
\draw[arrowr](e12) -- (q2.south west);
\draw[arrowr](es1) -- (s1);
\draw[arrowr](es2) -- (s2);
\end{tikzpicture}
\caption{Representation of a simple decomposable error model on the repetition code. The circles represent the 3 qubits of the code. The green boxes represent the two stabilizer generators $S_1 = Z \otimes Z \otimes I$ and  $S_2 = I \otimes Z \otimes Z$. The noise model decomposes into channels that act independently, as illustrated by the red boxes. For example, the channel $\mathcal{N}_1$ applies an $X$ error to the first qubit with some probability $\theta^1_{X \otimes I \otimes I}$. $\mathcal{N}_2$ applies the error $X \otimes X \otimes I$ with probability $\theta^2_{X \otimes X \otimes I}$ and the error $I \otimes X \otimes I$ with probability $\theta^2_{I \otimes X \otimes I}$. $\mathcal{N}_4$ flips the outcome of the measurement of $S_1$ with probability $\theta^4_{(1,0)}$.}
\label{fig:ErrorModelExample}
\end{figure}
\begin{definition}[decomposable error model]$ $ \newline
\label{def:DecomposablePauliChannel}
Let ${N_1,\dots,N_m \subset \errset}$ be disjoint error sets and ${I \notin N_i, \: \forall i \in \{1,\dots,m\}}$. For each ${i \in \{1,\dots,m\}}$, let ${\btheta^i = (\pei)_{\bm e \in N_i \cup \{I\}}}$
be a probability vector over $N_i \cup \{I\}$, and define ${\btheta = (\pei)_{i \in \{1,\dots,m\}, \bm e \in N_i}}$ by grouping together all these probability vectors and excluding the rates of trivial errors.
An error model is \emph{decomposable} with error sets $N_1,\dots,N_m$ and parameters $\bm \theta$ if errors from the different sets occur independently, i.e. the probability of a given error combination ${\bm X \in ((N_1 \cup \{I\}) \times \dots \times (N_m \cup \{I\}))}$ is \begin{equation}
\p{\bm X} = \prod_{i=1}^m \prod_{\bm e \in N_i} (\peia{\bm e}{i})^{\kron{X_i}{\bm e}}(\peia{I}{i})^{\kron{X_i}{I}}\, ,
\label{eq:p(X)}
\end{equation}
where $\delta$ is the Kronecker-delta and $\theta^i_I = 1 - \sum_{\bm e \in N_i} \theta^i_{\bm e}$.
\end{definition}
An example of such a model is given in \autoref{fig:ErrorModelExample}. 
There, the error sets would be $N_1 = \{X \otimes I \otimes I\},N_2 = \{X \otimes X \otimes I, I \otimes X \otimes I\},N_3 = \{I \otimes I \otimes X\},N_4 = \{(1,0)\},N_5 = \{(0,1)\}$.
 (Since errors here either only act on data qubits or only on syndrome bits we omitted the other trivial part of the errors.)
 For independent single-qubit Pauli noise the error sets would be given by $N_i = \{\bm X\exti{i},\bm Z\exti{i},\bm Y\exti{i}\}$.
 We will refer to the elements of the individual error sets as \emph{elementary errors}.
 Since there can be overlap between the supports of the different error channels, we often consider the vector $\bm X \in ((N_1 \cup \{I\}) \times \dots \times (N_m \cup \{I\}))$, containing all the elementary errors that occurred.
 The combined error $\bm E \in \mathsf P_n$ on the qubits and syndrome bits is then the product of all elementary errors that occurred, i.e. $\bm E = \prod_i X_i$.
 For independent single qubit Pauli noise $\bm X$ and $\bm E$ coincide.
 The map $\Mmap: \btheta \mapsto (\p{\Sb})_{\Sb \in \Synset}$ introduced in \autoref{sec:GeneralIdentifiabilityConditions} can now be written as
\begin{equation}
\p{\bm S} = \sum_{\bm X : \bm S(\bm X) = \bm S} \p{\bm X} \; .
\end{equation}
 Our identifiability conditions can now be straightforwardly generalized by considering the elementary errors as the new "single qubit errors". We also note that in the presence of measurement errors, it might be appropriate to include redundant stabilizer measurements such that the length $l$ of a syndrome is larger than the number of stabilizer generators \cite{Fujiwara_DataSyndromeCodes,Ashikhmin_QuantumDataSyndromeCodes,Delfosse_BeyondSingleShotFaultTolerance}. Our results also apply to such a scheme.
\subsubsection{Proof of Proposition \ref{lemma:LowRatesIdent}}
\label{sec:ProofLowRatesIdent}
Explicitly, \autoref{lemma:LowRatesIdent} is generalized as follows:
\begin{proposition}
\label{lemma:LowRatesIdent-general}
Consider a quantum code subject to a decomposable error model with error sets $N_1,\dots,N_m$ and parameters $\bm \theta$. Then the parameters of the channel are locally identifiable at $\bm \theta = \bm 0$ if and only if $\bm S (\bm e) \neq \bm S(\bm e \primed)$ for every choice of two different elementary errors $\bm e, \bm e\primed \in \bigcup_{i=1}^m N_i$.
\end{proposition}
\begin{proof}
We have to show that the map $\Mmap$ defined in \autoref{sec:IdentifiabilityConditions} is locally invertible at $\bm 0$.
The probability of the error $\Eb \in \errset$ is
\begin{equation}\label{eq:P[E]}
\p{\bm E} = \sum_{\substack{\bm X : \bm E(\bm X) = \bm E}} \p{\bm X} \; ,
\end{equation}
where $\p{\bm X}$ is given in \eqref{eq:p(X)}.
The probability of observing syndrome $\Sb$ is
\begin{equation}
\p{\bm S} = \sum_{\bm E \in \errset : \bm S(\bm E) = \bm S} \p{\bm E} \, .
\end{equation}
Thus the map $\Mmap$ decomposes as $\Mmap = \Smap \circ \Emap$, where
\begin{equation}
\begin{aligned}
&\Emap: \bm \theta \mapsto (\p{\bm E})_{\bm E \in \errset} \, ,
\end{aligned}
\end{equation}
describes the distribution of total errors, and
\begin{equation}
\begin{aligned}
&\Smap: (\p{\bm E})_{\bm E \in \errset} \mapsto (\p{\bm S})_{\bm S \in \Synset}  \, .
\end{aligned}
\end{equation}
describes the probability of each syndrome. Since $\Smap$ is linear, we have
\begin{equation}
\begin{aligned}
J \coloneqq \dtheta \Mmap &= D_{\Emap(\bm \theta)}\Smap \circ \dtheta \Emap \\
&= \Smap \circ \dtheta \Emap \, .
\end{aligned}
\end{equation}

We begin by calculating the derivative of $\p{\bm X}$.
\begin{equation}
\partialderiv{\p{\bm X}}{\pei}
= \kron{X_i}{\bm e}\p{\bm X_{-i}} - \kron{X_i}{I}\p{\bm X_{-i}} \, ,
\end{equation}
where $\bm X_{-i}$ denotes $\bm X$ without the $i$th component. Since we consider $\btheta = 0$, $\p{\bm X_{-i}}$ is zero if $X_j \neq I$ for any $i \neq j$. Thus
\begin{align}
\partialderiv{\p{\bm X}}{\pei} = \begin{cases} +1, & X_i = \bm e \text{ and } X\resti{j} = I \; \forall j \neq i \\ -1, & X_i = I \text{ and } X\resti{j} = I \; \forall j \neq i \\ 0, & \text{otherwise} \end{cases}
\end{align}
We then have
\begin{equation}
\begin{aligned}
D_{(\btheta = \bm 0)} \Emap[\bm E, \pei] &= \partialderiv{\p{\bm E}}{\pei} = \sum_{\bm X : \bm E(\bm X) = \bm E} \partialderiv{\p{\bm X}}{\pei}  \\
&= \begin{cases} +1, & \bm E = \bm e \\ -1, & \bm E = I \\ 0, & \text{otherwise} \end{cases}
\end{aligned}
\end{equation}
where the last line follows because there is always at most one non-zero summand, since the different error sets are by definition disjoint. Therefore the derivative of $\Emap$ has a very simple form:
\begin{equation}
D_{(\btheta = \bm 0)} \Emap = 
\begin{pmatrix}
-1 & \dots & -1\\
 \bm u_{\bm e_1} & \dots & \bm u_{\bm e_k}
\end{pmatrix} \, ,
\end{equation}
where $\bm u_{\bm e_i}$ denotes the unit vector associated to the corresponding elementary error $\bm e_i$, and ${k = \sum_{i = 1}^m |N_i|}$. Since the error sets $N_1,\dots,N_m$ are disjoint, i.e.\ there are no duplicate elementary errors, this matrix has $k$ independent columns and thus full column rank. As long as no two elementary errors have the same syndrome, the images of these columns under $\Smap$ are again linearly independent. Then $D_{(\btheta = \bm 0)} \Mmap = \Smap \circ D_{(\btheta = \bm 0)} \Emap$ has full column rank, and the inverse function theorem completes the proof.
\end{proof}

\subsubsection{Proof of Theorem \ref{lemma:GeneralIdentifiability}}
\label{sec:ProofGeneralIdentifiability}
Our general version of \autoref{lemma:GeneralIdentifiability} is:
\begin{theorem}
\label{lemma:GeneralIdentifiability-general}
Consider a quantum code subject to a decomposable error model with error sets $N_1,\dots,N_m$ and parameters $\bm \theta$.
Assume that $\theta^i_{\bm e} > 0$ for all $i$ and all $\bm e \in N_i$, and that $\p{\bm S} > 0$ for all syndromes $\bm S \in \Synset$. 
Then the error rates are locally identifiable at $\bm \theta$ if and only if the matrix $\tilde{J}$ with entries
\begin{equation}
\tilde{J}[\bm S, \pei] = \frac{\p{X_i = \bm e | \bm S}}{\p{X_i = \bm e}} - \frac{\p{X_i = I | \bm S}}{\p{X_i = I}}
\end{equation}
has full column rank.
\end{theorem}
\begin{proof}
We have to show that the map $\Mmap$ defined in \autoref{sec:IdentifiabilityConditions} is locally invertible at $\btheta$.
Since we assume that the rates of all errors and syndromes are strictly greater than $0$, the $\Mmap$ will be locally invertible at $\bm \theta$ if and only if the entry-wise logarithm $\ln(\Mmap)$ is locally invertible at $\bm \theta$.
 Thus we consider the derivative of the log-likelihood $\ln(\p{\bm S})$ for each syndrome.
 Remember that the probability of a syndrome $\bm S$ can be expressed as
\begin{equation}
\p{\bm S} = \sum_{\bm X  : \bm S(\bm X) = \bm S} \p{\bm X} \, ,
\end{equation}
where $\p{\bm X}$ is given in \eqref{eq:p(X)}.
 As in the proof of \autoref{lemma:LowRatesIdent} we compute the derivative
\begin{align*}
\partialderiv{\p{\bm X}}{\pei}
&= \kron{X_i}{\bm e}\p{\bm X_{-i}} - \kron{X_i}{I}\p{\bm X_{-i}} \\
&= \kron{X_i}{\bm e}\frac{\p{\bm X}}{\pei} - \kron{X_i}{I}\frac{\p{\bm X}}{\peia{I}{i}} \\
&= \kron{X_i}{\bm e}\frac{\p{\bm X}}{\p{X_i = \bm e}} - \kron{X_i}{I}\frac{\p{\bm X}}{\p{X_i = I}} \, .
\end{align*}
Using the fact that
\begin{equation}
\p{\bm S} = \sum_{\bm X | \Sb(\bm X) = \Sb} \p{\bm X} \, ,
\end{equation}
we obtain
\begin{equation}
\begin{aligned}
&\partialderiv{\ln(\p{\bm S})}{\pei} \\
&= \frac{1}{\p{\bm S}} \sum_{\bm X  : \bm S(\bm X) = \bm S} \left( \kron{X_i}{\bm e}\frac{\p{\bm X}}{\p{X_i = \bm e}} - \kron{X_i}{I}\frac{\p{\bm X}}{\p{X_i = I}} \right) \\
&= \frac{1}{\p{\bm S}} \left(  \frac{1}{\p{X_i = \bm e}}\sum_{\bm X  : \bm S(\bm X) = \bm S}\kron{X_i}{\bm e}\p{\bm X} - \frac{1}{\p{X_i = I}}\sum_{\bm X  : \bm S(\bm X) = \bm S}\kron{X_i}{I}\p{\bm X} \right) \\
&= \frac{\p{X_i = \bm e | \bm S}}{\p{X_i = \bm e}} - \frac{\p{X_i = I | \bm S}}{\p{X_i = I}} \, .
\end{aligned}
\label{eq:LLDerivative}
\end{equation} 
By the inverse function theorem, this completes the proof.
\end{proof}

\subsection{Proof of Lemma \ref{lemma:PerfectCodeEqualCondProbs}}
\label{sec:ProofPerfectCodeEqualCondProbs}
We will now proof \autoref{lemma:PerfectCodeEqualCondProbs} in order to finish the proof of \autoref{lemma:PerfectCodeIdent}. Remember that the $i$-th tensor factor of $\bm E \in \mathsf P_n$ is denoted $E\resti{i}$. Furthermore, the Pauli acting as $e \in P_1$ on qubit $i$ and as the identity everywhere else is denoted $\bm e\exti{i}$. Finally, for $\bm E \in \mathsf P_n$, we use $\Ebi{i}$ as a shorthand for $(E\resti{i})\exti{i}$.
We define the weight of a Pauli error in the standard way.
\begin{definition}[weight]
The \emph{weight} of a Pauli error $\Eb = \Ei{1} \otimes \Ei{2} \otimes \dots \otimes \Ei{n} \in \mathsf P_n$ 
is defined as
\begin{equation}
\wt(\Eb) \coloneqq |\{ \Ei{i} \mid \Ei{i} \neq I, i \in \{1,\dots,n\}\}| \, .
\end{equation}
\end{definition}
In the case of equal error rates $p$, the probability of an error is determined by its weight. Let us denote $\pbar \coloneqq 1 - p$. We obtain a convenient expression for $\PeiS{e}$. 
For $e \neq I$ we have
\begin{equation}
\begin{aligned}
\PeiS{e}
&= \sum_{\substack{\bm E \in \mathsf P_n:\\ \bm S(\bm E)= \bm S, E_i = e}} \p{\Eb}
\\
&= \sum_{\substack{\bm E \in \mathsf P_n:\\ \bm S(\bm E)= \bm S, E_i = e}} p^{\wt(\Eb)}(\pbar)^{n-\wt(\Eb)} \\
&= p\sum_{\substack{\bm E \in \mathsf P_n:\\ \bm S(\bm E)= \bm S, E_i = e}} p^{\wt(\Eb_{-i})}(\pbar)^{n-1-\wt(\Eb_{-i})}
\end{aligned}
\label{eq:p(e_i,S)_e}
\end{equation}
and, analogously, for $e = I$ 
\begin{equation}
\PeiS{e} = \pbar\sum_{\substack{\bm E \in \mathsf P_n:\\ \bm S(\bm E)= \bm S, E_i = e}} p^{\wt(\vec E_{-i})}(\pbar)^{n-1-\wt(\vec E_{-i})} \, .
\label{eq:p(e_i,S)_I}
\end{equation}
By the definition of conditional probability we obtain
\begin{equation}
\pc{\bm S}{E\resti{i} = e} = \sum_{\substack{\bm E \in \mathsf P_n:\\ \bm S(\bm E)= \bm S, E_i = e}} p^{\wt(\vec E_{-i})}(\pbar)^{n-1-\wt(\vec E_{-i})} \, .
\end{equation}
\autoref{lemma:PerfectCodeEqualCondProbs} is thus equivalent to the following lemma:
\begin{lemma}
\label{lemma:PerfectCodeEqualSums}
Consider a perfect single error correcting code on $n$ qubits. Let $e,e\primed \in \mathsf P_1$. Then for any syndrome $\Sb \in \Synset \setminus \{\bm 0\}$, error rate $p \in [0,1]$ and qubit $i$ such that $\Sb \neq \Sb(\eei{i})$ and $\Sb \neq \Sb((\bm e\primed)\exti{i})$
the following equality holds:
\begin{equation}\label{eq:sums_eq}
\sum_{\substack{\bm E \in \mathsf P_n:\\ \bm S(\bm E)= \bm S, E_i = e}}p^{\wt(\Eb_{-i})}(\pbar)^{n-1-\wt(\Eb_{-i})} = \sum_{\substack{\bm E \in \mathsf P_n:\\ \bm S(\bm E)= \bm S, E_i = e\primed}} p^{\wt(\Eb_{-i})}(\pbar)^{n-1-\wt(\Eb_{-i})}
\end{equation}
\end{lemma}

In other words, we have to show that the sums in the expression do not depend on $e$ except if $\Sb = \Sb(\eei{i})$. This is the case if for all $w = 0,\dots,n-1$ and $e,e\primed \in \mathsf P_1$ such that $\Sb \neq \Sb(\e\exti{i}),\Sb((\bm \e\primed)\exti{i})$ we have
\begin{equation}
|\{\bm E \mid E\resti{i} = e,\wt(E\wi{i}) = w, \Sb(\Eb) = \Sb\}| \overset{!}{=} |\{\bm E \mid E\resti{i} = e\primed,\wt(E\wi{i}) = w, \Sb(\Eb) = \Sb\}| \, ,
\end{equation}
since then the coefficients for each of the exponents appearing in the expressions will be equal. Therefore, in the following, we will derive an expression for the ``modified'' weight distribution
\begin{equation}\label{eq:def:kw}
k_w(e,i,\Sb) \coloneqq |\underbrace{\{\Eb \mid E\resti{i} = e,\wt(\Eb\wi{i}) = w, \Sb(\Eb) = \Sb\}}_{\coloneqq K_w(e,i,\Sb)}|
\end{equation}
We will show that this distribution is independent of $e$ if $\Sb \neq \Sb(\bm e\exti{i})$. For the rest of this section, we fix a qubit $\qht \in \{1,\dots,n\}$, an error $\eht \in \mathsf P_1$ which will act on $\qht$ and some syndrome $\bm 0\neq \Sst \in \mathbb{F}_2^l$, and we denote $k_w \coloneqq k_w(\eht,\qht,\Sst)$ and $K_w \coloneqq K_w(\eht,\qht,\Sst)$. For now, we do not assume that $\Sst \neq \Sb(\bm \eht\exti{\qht})$.

\begin{notation}[Perfect Code Property]
Since we consider a perfect single error correcting code, for each syndrome $\Sb$ there exists a unique single qubit error $\eeq$ with $\Sb (\eeq) = \Sb$. 
We denote this error by $\Sb^{-1}(\Sb)$.
\end{notation}

The core idea of the proof is to construct the sets $K_w$ iteratively. 
We can use the perfect code property to construct weight $w$ errors with syndrome $\Sst$ from weight $w-1$ errors with any syndrome $\Sb\primed$ by adding the unique single qubit error  $\Sb^{-1} (\Sb\primed \oplus \Sst)$. 
We formalize this as follows.

\begin{definition}[\Modst{} and \Extst{}]
Let $\bm E \in \mathsf P_n$. 
We say an error $\Est$ is a \emph{\Modst{}} of $\Eb$ if $\Sb (\Est) = \Sst$ and there exists a single qubit error $\eeq$  with $\Est = \Eb \eeq$. 

We say $\Est$ is an \emph{\Extst{}} of $\Eb$ if $\Est$ is a \Modst{} of $\Eb$ with $\wt(\Est\wqht) = \wt(\Ewqht) + 1$ and $\Estqht = \eht$.
\end{definition}

Note that this definition does depend on the choice of $\vec \eht^{(q)}$ and $\Sst$, 
which is fixed for the rest of this section. 

It is simple to construct a \Modst{} for each error.

\begin{lemma}\label{lem:unique_modification}
Each error $\Eb \in \mathsf P_n$ has a \textbf{unique} \Modst{}. 
We denote it $\Est$.
\end{lemma}
\begin{proof}
Let $\eeq = \Sb^{-1}(\Sst \oplus \Sb (\bm E))$. Then $\Eb \eeq$ is a \Modst{} of $\bm E$. Furthermore, for two possible \Modst{}s $\Eb \eeq, \Eb (\bm e\primed)\exti{q\primed}$ with $\Sb (\Eb \eeq) = \Sb (\Eb (\bm e\primed)\exti{q\primed}) = \Sst$, we obtain $\Sb (\eeq) = \Sb ((\bm e\primed)\exti{q\primed}) = \Sst \oplus \Sb (\Eb)$. Because we consider a perfect code this implies $\eeq = (\bm e\primed)\exti{q\primed}$. Thus the \Modst{} is unique. 
\end{proof}

However, it is possible that an error $\Eb$ does not have a \Extst. This happens for example if the unique single qubit error that needs to be added to obtain the \Modst{} is already in $\Eb$, or if it is on $\qht$. 
We formalize this in the following corollary. 

\begin{corollary}
\label{lemma:NoStarExtension}
Let $\Eb \in \mathsf P_n$ be an error with $\Eqht = \eht$ and $\wt(\Eb\wqht) = w$. $\Eb$ does \textbf{not} have an \Extst{} if and only if one of the following mutually exclusive conditions is true:
\begin{enumerate}[label=(\roman*)]
\item \label{item:1} $\Est = \Eb$
\item \label{item:2} $\wt(\Est\wqht) = w-1 \wedge \Estqht = \eht$
\item \label{item:3} $\Est \neq \Eb \wedge \wt(\Est\wqht) = w \wedge \Estqht = \eht$
\item \label{item:4} $\Estqht \neq \eht$
\end{enumerate}
where as always $\Est$ is the unique \Modst{} of $\Eb$.
If we write $\Est = \Eb \eeq$, where $\eeq$ is a uniquely determined single qubit error acting on qubit $q$, these conditions are equivalent to
\begin{enumerate}[label=(\roman*')]
\item $e = I$
\item $\Ei{q} = e \wedge q \neq \qht \wedge e \neq I$ 
\item $\Ei{q} \neq e \wedge \Ei{q} \neq I \wedge q \neq \qht \wedge e\neq I$
\item $q = \qht \wedge e \neq I$
\end{enumerate}
 
\end{corollary}

\begin{proof}
By definition $\Est = \bm E \eeq$ is an \Extst{} of $\Eb$ if and only if 
\begin{equation}
\begin{aligned}
 & \Estqht = \eht \ \wedge\ \wt(\Est\wqht) = w + 1
\\
  \Leftrightarrow \quad  &  q \neq \qht \ \wedge\ E\resti{q} = I \ \wedge\ e \neq I \, ,
\end{aligned}
\end{equation}
where we have used that $\Eqht = \eht$. 
Negating this statement and using that ${\wt(\Est\wqht) \in \{w-1,w,w+1\}}$ leads to the conditions above. 
\end{proof}

Since similar reasoning will be used repeatedly throughout this section, let us illustrate some of the cases in \autoref{lemma:NoStarExtension} with an example.
 Consider the 5 qubit perfect code with stabilizer generators $g_1 = X \otimes Z \otimes Z \otimes X \otimes I$, $g_2 = I \otimes X \otimes Z \otimes Z \otimes X$, $g_3 = X \otimes I \otimes X \otimes Z \otimes Z$, and $g_4 = Z \otimes X \otimes I \otimes X \otimes Z$. 
 For this example, let $\qht = 1$, $\eht = X$, and  $\Sst = (1,0,0,1)$.
 The error $\Eb = X \otimes X \otimes X \otimes I \otimes I$ has the syndrome $(0,1,0,1)$, and thus its \Modst{} is obtained by applying $\Sb^{-1}((1,1,0,0)) = \bm X\exti{3}$, resulting in $\Est = X \otimes X \otimes I \otimes I \otimes I$.
 This is not a valid \Extst{} since the weight was reduced, corresponding to case \ref{item:2} in \autoref{lemma:NoStarExtension}.
 The single qubit error we applied canceled with an existing error in $\Eb$.
 On the other hand, the error $\Eb = X \otimes I \otimes Z \otimes Z \otimes I$ has the syndrome $\Sb(\Eb) = (1,0,1,0)$. 
 Thus its \Modst{} is obtained by adding $\bm e = \Sb^{-1}((0,0,1,1)) = \bm X\exti{5}$, resulting in $\Est = X \otimes I \otimes Z \otimes Z \otimes X$. 
 This is a valid \Extst{}. Notice that the additional single qubit error was applied on qubit 5 where $\Eb$ acts trivially, or equivalently,  $E^{\ast}\resti{5} = e$.
 
In \autoref{lemma:NoStarExtension} we categorized errors without a valid \Extst{} by their \Modst{}.
Now we characterize $k_w$ in terms of \Extst{}s. 

\begin{lemma}
\label{lemma:numberkstar}
For any $w > 0$
\begin{align*}
k_w = |\{\Eb \in \mathsf P_n \mid \exists \Eb\primed \in \mathsf P_n: \Eb \text{ is a \Extst{} of } \Eb\primed, E\primed\resti{\qht} = \eht ,\wt(\Eb\primed\wqht) = w - 1\}| \, .
\end{align*}
\end{lemma}

\begin{proof}
By definition of $k_w$ \eqref{eq:def:kw}, we have to show that
\begin{align*}
&|\{\Eb \in \mathsf P_n \mid \Eqht = \eht,\wt(\Ewqht) = w, \Sb(\Eb) = \Sst\}| \\
=&|\{\Eb \in \mathsf P_n \mid \exists \Eb\primed \in \mathsf P_n: \Eb \text{ is a  \Extst{} of } \Eb\primed, E\primed\resti{\qht} = \eht ,\wt(\Eb\primed\wqht) = w - 1\}| \, .
\end{align*}

$\supseteq$: By definition of \Extst{}.

$\subseteq$: Let $\Eb \in \mathsf P_n$ be an error such that $\Eqht = \eht,\wt(\Ewqht) = w$ and $\Sb(\Eb) = \Sst$. Chose a qubit $q \neq \qht$ such that $\Ei{q} \neq I$.
Then $\Eb$ is an \Extst{} of $\Eb\primed \coloneqq \Eb \Ebi{q}$. Furthermore $\wt(\Eb\primed\wqht) = \wt(\Ewqht) - 1$ and $E\primed\resti{\qht} = \eht$ by definition of $\Eb\primed$.
\end{proof}

While this establishes a connection between the weight distribution $k_w$  and the concept of \Extst{}, it is difficult to count all errors that are valid \Extst{}s. 
A number easier to characterize is
\begin{equation}
l_w \coloneqq |\underbrace{\{\Eb \in \mathsf P_n \mid \Eb \text{ has a \Extst }, \Eqht = \eht,\wt(\Eb\wqht) = w-1\}}_{\coloneqq L_w}| \, .
\end{equation}
This is similar to the characterization in \autoref{lemma:numberkstar}, but $l_w > k_w$ because two different errors can have the same \Extst. We have to correct for this ``double counting''. 

\begin{lemma}
\label{lemma:kVslRelation}
\begin{equation}
k_w = \frac{l_w}{w}
\end{equation}
\end{lemma}
\begin{proof}
By \autoref{lem:unique_modification} we have a well defined function $g: \mathsf P_n \mapsto \mathsf P_n$ that maps an error $\Eb \in \mathsf P_n$ to its \Modst{} $\Est \in P_n$. By \autoref{lemma:numberkstar} and the definition of $L_w$, $g$ maps $L_w$ to $K_w$, and the restriction $g_{|L_w}: L_w \rightarrow K_w$ is surjective. Because the \Modst{} is unique, the pre-images of two distinct elements of $K_w$ under $g$ are disjoint. Thus,
\begin{equation}
|L_w| = \sum_{\Eb \in K_w} |g_{|L_w}^{-1}(\Eb)|
\end{equation}
We want to determine the size of these pre-images. So let $\Eb \in K_w$. From the definition of $L_w$ and the definition of \Extst{}, it follows that $\Eb\primed \in g_{|L_w}^{-1}(\Eb)$ if and only if there exists a qubit $q \neq \qht$ such that $E\resti{q} \neq I$ and $\Eb\primed = \Eb \Ebi{q}$. Thus, since by definition $\wt(\Eb\wqht) = w$, the pre-image has $w$ elements. This concludes the proof.
\end{proof}

Thus, we can characterize the weight distribution $k_w$ through the numbers $l_w$, for which we derive a recursive formula. 

\begin{lemma}
\label{lemma:WeightDistEquation}
There exists a recursive formula relating $k_w$ to $k_{w-1}$ and $k_{w-2}$
\end{lemma}

\begin{proof}
We prove that 
\begin{equation}
\begin{split}
l_w =  3^{w-1}\binom{n-1}{w-1} &- k_{w-1}
\\
&
- 3(n-w+1)k_{w-2} 
  \\
  &
- 2(w-1)k_{w-1}  
\\
&
- \sum_{\substack{e\primed \in \mathsf{P}_1 \\ \eht  \neq e\primed}} k_{w-1}(e\primed,\qht,\Sst)
\end{split}
\label{eq:Recursivel_w}
\end{equation}
for any $2\leq w\leq n$. \autoref{lemma:kVslRelation} then gives the corresponding equation for $k_w$.

The total number of errors $\Eb \in \mathsf P_n$ with $\wt(\Ewqht) = w-1$ and $\Eqht  = \eht$ is $3^{w-1}\binom{n-1}{w-1}$ since there are $\binom{n-1}{w-1}$ ways to chose $w-1$ positions in $n-1$ positions, and 3 possible Paulis on each position. 
Next we count how many of them do \textbf{not} have an \Extst. The different conditions for this are given in \autoref{lemma:NoStarExtension}, where errors without an \Extst{} are categorized by their \Modst{}. We count the number of errors $\Eb \in \mathsf P_n$ with $\wt(\Ewqht) = w-1$ and $\Eqht  = \eht$ fulfilling each of these different conditions. We can group errors without a valid \Extst{} by their \Modst{}, i.e.\ 
\begin{equation}\label{eq:non-extensions}
\begin{split}
&\{\Eb \in \mathsf P_n \mid \wt(\Ewqht) = w-1,\Eqht = \eht, \Eb \text{ does not have a \Extst{} } \} \\
= 
&\underset{\substack{\Eb\primed \in \mathsf P_n : \\ \Eb\primed \text{ is not a \Extst}}}{\bigcup} \{\Eb \in \mathsf P_n \mid \Eb^{\ast} = \Eb\primed, \Eqht =\eht, \wt(\Ewqht) = w-1 \} \, ,
\end{split}
\end{equation}
where all the individual sets are disjoint because the \Modst{} is unique.
To do this, we have to consider the following cases, for each of which $\wt(\Ewqht) = w-1$ and $\Eqht  = \eht$ hold. 

\begin{description}[leftmargin=\parindent,labelindent=0pt,font=\normalfont]
\item[Case \ref{item:1}] $\Est = \Eb$.\\
This condition is equivalent to $\Sst = \Sb(\Eb)$.
By definition there are $k_{w-1}$ such errors.

\item[Case \ref{item:2}] $\wt(\Est\wqht) = w-2 \wedge \Estqht = \eht$. \\
For each error $\Eb$ fulfilling this condition, we have that $\Eb = \Est \eeq$ for a Pauli $e\in \mathsf P_1 \setminus \{I\}$ and a qubit $q \neq  \qht$ with $\Estq = I$.
For a given error $\Eb\primed$ with $\wt(\Eb\primed\wqht) = w - 2$, 
$\Eb\primed_{\qht} = \eht$ and $\Sb(\Eb\primed) = \Sst$, 
there are $n-1-(w-2) = n-w+1$ possibilities to chose a qubit $q \neq \qht$ with $E\primed\resti{q} = I$. 
For each of these, there are 3 different Paulis one could add to this position.
Each of these gives a distinct error $\Eb$ with $\Est = \Eb\primed$. 
The total number of errors $\Eb\primed$ with $\wt(\Eb\primed\wqht) = w - 2$, $\Eb\primed_{\qht} = \eht$ 
and $\Sb(\Eb') = \Sst$ 
is by  definition $k_{w-2}$, and because the \Extst{} is unique they all give distinct contributions. 
Thus,
\begin{align*}
&|\{\Eb \in \mathsf P_n \mid \wt(\Est\wqht) = w-2, \Estqht = \eht, \wt(\Ewqht) = w-1, \Eqht  = \eht\}|  \\ 
&= 3(n-w+1)\,|\{\Eb\primed \in \mathsf P_n \mid E\primed\resti{\qht} = \eht, \wt(\Eb\primed\wqht) = w-2, \Sb(\Eb\primed) = \Sst \}| \\
& = 3(n-w+1)k_{w-2} \, .
\end{align*}

\item[Case \ref{item:3}:] $\Est \neq \Eb \wedge \wt(\Est\wqht) = w-1 \wedge \Estqht = \eht$. \\
For each such error $\Eb$ it holds $\Eb = \Est\eeq$ for a Pauli $e\in \mathsf P_1 \setminus \{I, \Estq\}$ and a qubit $q \neq \qht$ with $\Estq \neq I$. 
For a given error $\Eb\primed$ with ${\wt(\Eb\primed\wqht) = w - 1}$, there are $w-1$ choices for $q \neq \qht$ such that $E\primed\resti{q}  \neq I$, and for each choice of $q$ there are 2 possible choices of ${e \in P_1 \setminus \{I,E\primed\resti{q}\}}$. 
The total number of errors $\Eb\primed$ with $ \wt(\Eb\primed\wqht) = w - 1, E\primed\resti{\qht} = \eht$ and 
$\Sb(\Eb\primed) = \Sst$ is by definition $k_{w-1}$, and again they give distinct contributions.
Thus,
\begin{align*}
&|\{\Eb \in \mathsf P_n \mid \Est \neq \Eb, \wt(\Est\wqht) = w-1, \Estqht = \eht,  \wt(\Ewqht) = w-1, \Eqht  = \eht\}| \\ 
&= 2(w-1)|\{\Eb\primed \mid E\primed\resti{\qht} = \eht, \wt(\Eb\primed\wqht) = w-1, \Sb(\Eb\primed) = \Sst \}| \\
& = 2(w-1)k_{w-1}\, .
\end{align*}

\item [Case \ref{item:4}:] $\Estqht \neq \eht$. \\
For each such error $\Eb$ there exists a corresponding $\Eb\primed = \Est$ such that $\Eb = \Eb\primed \eei{\qht}$ for an appropriate Pauli $e \in \mathsf P_1 \setminus \{I\}$.
Note that $\wt(\Eb\primed\wqht) = \wt(\Eb\wqht) = w-1$.
The total number of errors $\Eb\primed$ with $\wt(\Eb\primed\wqht) = w-1$, $E\primed\resti{\qht} \neq \eht$ and $\Sb(\Eb\primed) = \Sst$ is by definition $\sum_{e\primed \in \mathsf P_1 | \eht  \neq e\primed} k_{w-1}(e\primed,\qht,\Sst)$, and because the \Modst{} is unique the different $e\primed$ give different contributions. 
\end{description}

There is no double counting because the union in \eqref{eq:non-extensions} is disjoint.
Finally we obtain the number of errors that have a valid \Extst{} by subtracting the number of errors without a valid \Extst{} from the total number of errors, which yields the recursion \eqref{eq:Recursivel_w}.
\end{proof}

With this recursive formula we can easily prove by induction that for a given qubit $q$, $k_w(e,q,\Sb)$ is (almost) independent of $e$ and $\Sb$.

\begin{proof}[Proof of \autoref{lemma:PerfectCodeEqualSums}]
We consider again a fixed qubit $\qht$ and syndrome $\Sst \in \Synset$, 
and prove that the numbers $k_w(e,\qht,\Sst)$ are equal for any $e \in \mathsf P_1$ such that $\Sb(\eei{\qht}) \neq \Sst$. 
Let $\eht \in \mathsf P_1$ with $\Sb((\bm \eht)\exti{\qht}) \neq \Sst$. We consider two different cases, corresponding to different initial conditions for 
\autoref{lemma:WeightDistEquation}. 
The two cases are:
\begin{enumerate}
\item $\Sst = \Sb((\bm e\primed)\exti{\qht})$ for some $e\primed \neq \eht$
\item $\Sst \neq \Sb(\bm e\exti{\qht})$ for any error $e$ acting on qubit $\qht$
\end{enumerate}
Consider case 2 first.
For $w = 0$, we have $k_0(\eht,\qht,\Sst) = 0$ independent of $(\eht,\Sst)$ because $\Sst \neq \Sb(\bm e\exti{\qht}) \, \forall e \in P_1$.
For $w = 1$, $k_1(\eht,\qht,\Sst) = 1$ is independent of $(\eht, \Sst)$ because the only error $\eeq$ with $\Sb(\bm \eht\exti{\qht}\eeq) = \Sst$ is $\Sb^{-1}(\Sb(\bm \eht\exti{\qht}) \oplus \Sst))$ (and this error does not act on $\qht$ because $\Sst \neq \Sb(\bm e\exti{\qht}) \forall e \in P_1$.). For $w > 1$, the claim follows by induction since the right hand side of the recursive equation in \autoref{lemma:WeightDistEquation} is now independent of $\eht$ and $\Sst$. This concludes the proof for case 2. In case 1 the initial conditions are $k_0 = 0, k_1 = 0$. The rest of the proof is analogous. The only caveat is that the last term of \eqref{eq:Recursivel_w} now also contains a term $k_{w-1}(e\primed,\qht,\Sst)$ for an error $e\primed$ with $\Sb(\bm e\exti{\qht}) = \Sst$. But this term can be computed using the same recursive equation, and does not depend on $e$.
\end{proof}

This finally concludes the proof of \autoref{lemma:PerfectCodeEqualSums}, and thus also the proof of \autoref{lemma:PerfectCodeIdent}. As mentioned above, \autoref{lemma:WeightDistEquation} can also be used to calculate the numbers $k_w(e,\qht,\Sst)$ for the remaining case $\Sst = \Sb(\bm e\exti{\qht})$. The correct initial conditions are $k_0 = 1, k_1 = 0$.

\subsection{Additional Numerical Results}
\label{sec:AdditionalNumericalResults}
Here, we provide data complementary to the results shown in \autoref{sec:NumericalResults}. In particular, we consider the \acf{MSE} of the proposed estimator, and we show results with noisy measurements.
\subsubsection{\acs{MSE} of the Estimator}
First, we will demonstrate that the \acs{EM} estimator achieves the \acf{CRB}. The \acs{MSE} of the estimator $T$ of a parameter $\theta$ can expressed by the bias-variance decomposition
\begin{equation}
\operatorname{MSE} = \operatorname{bias}(T)^2 + \operatorname{var}(T) \, .
\end{equation}
Assume we want to estimate the error rates $\btheta$ of a code from $m$ independent syndrome observations. Then the covariance of any unbiased estimator $\bm T$ of $\btheta$ is bounded by the \ac{CRB}
\begin{equation}
\operatorname{cov}_{\btheta}(\bm T) \geq \frac{I(\btheta)^{-1}}{m}
\end{equation}
i.e.\ $\operatorname{cov}_{\btheta}(\bm T) - \frac{I(\btheta)^{-1}}{m}$ is a positive semi-definite matrix; 
 here, the Fisher information matrix $I$ is defined by 
\begin{equation}
I_{i,j} = \mathbb{E}_{\Sb}\biggl[\partialderiv{\ln(\p{\Sb | \btheta})}{\theta_i}\partialderiv{\ln(\p{\Sb | \btheta})}{\theta_j}\biggr] . 
\end{equation}
In particular, the variance in the estimate of a single parameter is bounded by the diagonal entries of the inverse of the Fisher information. The derivative $\partialderiv{\ln(\p{\Sb | \btheta}}{\theta^i_e}$ of the log-likelihood with respect to a parameter $\theta^i_e$ was already computed in \eqref{eq:LLDerivative} as
\begin{equation}
\partialderiv{\ln(p(\Sb | \btheta)}{\theta^i_e} = \frac{\PeicS{e}}{\Pei{e}} - \frac{\PeicS{I}}{\Pei{I}} \, .
\end{equation}
Since the probabilities $\p{e_i | \bm S}$ can be computed using \acs{BP}, we can numerically evaluate this bound for concrete codes and noise models and compare our estimator to this bound. 
However, for concatenation levels beyond the first, it was necessary to approximate the expectation value over all syndromes by Monte-Carlo sampling. We used $10^6$ samples to do this. As a side note, it is not sufficient to consider the \ac{CRB} for direct observation of the errors (which is much easier to evaluate). It can be shown that the Fisher information always decreases when post-processing the data, and thus the bounds for syndrome observations must necessarily be higher than for direct measurements of the errors (in our cases the difference was about a factor 2). 
Finally, it should be noted that in our simulations we have access to the actual error rates which makes it possible to compute the \ac{MSE}. 
In a real experiment, one could for example consider the variance instead. In our tests, the \ac{EM} estimator exhibited a squared bias that was very small compared to the variance, such that the variance coincides with the \ac{MSE}. However, the \ac{HEM} estimator showed significant bias in some settings. In the following, we always consider the \ac{MSE} in the estimation of $\theta^1_X$. However, plots for the other parameters look similarly. The \ac{MSE} was always determined over $10 ^3$ simulations for each data point. We consider the \ac{MSE} of the estimation at error rate $p=0.13$.
For a relatively bad initialization results were already shown in \autoref{fig:MSEOverConcat_A20_p013_nest1000}. Here, we consider the situation where an accurate initialization is available, demonstrated by using $\alpha = 200$. An example comparing the \ac{MSE} at the first concatenation level is shown in \autoref{fig:MSEOverDataSize_a200_p013_concat1}.
\begin{figure}[htb]
\includegraphics[scale=1]{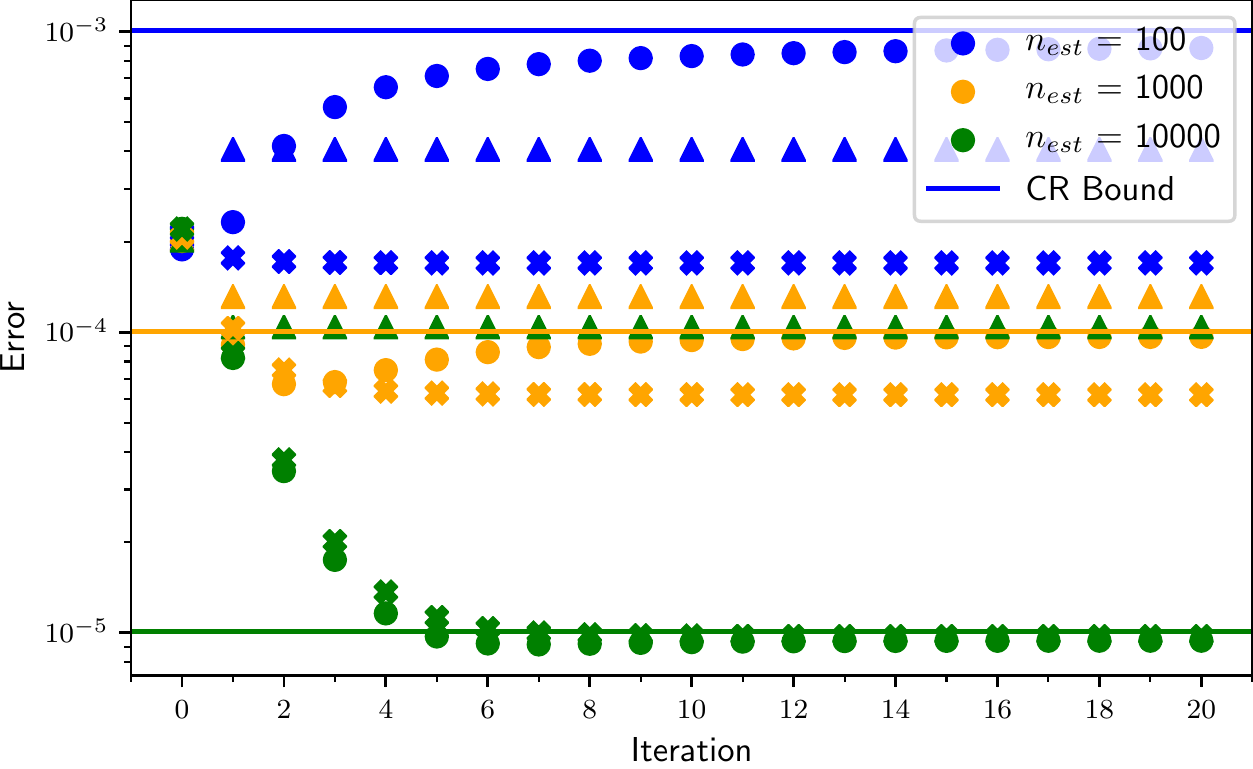}
\caption{Comparison of the \ac{MSE} in $\theta^1_X$ between \ac{EM} (circles), \ac{HEM} (triangles) and regularized \ac{EM} (crosses) for different amounts of estimation data $\nest$ for a good initialization at the first concatenation level. The parameters were $p = 0.13$, $\alpha= 200$, $\nconcat = 1$. $\beta = 200$ was used for the regularized version.
}
\label{fig:MSEOverDataSize_a200_p013_concat1}
\end{figure} 
For low data sizes, the initialization is more accurate than the estimate using the data set. In this case, \ac{HEM} outperforms \ac{EM} and even beats the \ac{CRB} (remember that the \ac{CRB} as it is used here only applies to unbiased estimators). The reason is that \ac{HEM} has a strong bias towards the initial parameters, which did not decrease with the size of the data sets or the number of iterations in our simulations. At larger data sizes this bias is detrimental, and it can be seen that \ac{EM} outperforms \ac{HEM}.
Especially for low numbers of iterations, \ac{EM} also exhibits some bias towards the initialization.
 This can be desirable in case of a good initialization, since it explains why \ac{EM} also slightly beats the \ac{CRB} at low data sizes. 
In particular, we see that at $\nest = 100$ and $\nest = 1000$ \ac{EM} performs better if a low number of around 3 iterations is used.
 Note that a small bias remains at higher iterations, which explains why \ac{EM} also slightly beats the \ac{CRB}.
Especially interesting is the case $\nest = 1000$, where \ac{EM} both improves over the initialization and clearly beats the \ac{CRB} at low numbers of iterations.
Since we do not know beforehand after how many iterations the procedure should be stopped, it is sensible to instead regularize the estimator in such a setting, such that it does not converge away from the improved value at low iterations. The regularization is done by introducing a Dirichlet prior
\begin{equation}
\p{\btheta^i} = \frac{1}{B(\alpha)} \prod_{e \in \mathsf P_1} (\theta^i_e)^{\beta^i_e}
\end{equation}
over the initialization, representing information on its accuracy (see Ref.~\cite{Bishop_ML}). Here, $\beta^i_e = (1 - (\theta^{(0)})^i_e)\beta$ and the real hyper-parameter $\beta$ controls the strength of the regularization.
The effect of this regularization, using $\beta = 200$, is also demonstrated in \autoref{fig:MSEOverDataSize_a200_p013_concat1} (the cross-shaped markers). It can be seen that the regularized \ac{EM} algorithm converges roughly to the minimum of the unregularized version, which was the desired effect. For large data sizes, the regularization introduces a minimal increase in the estimation error. We also tested regularizing the \ac{HEM} version in the same manner, but no improvements were obtained. Similar results could be obtained for higher concatenation levels.
 The main difference is that \ac{HEM} performs worse at higher levels.

\subsubsection{Estimator with Measurement Noise}

We consider a phenomenological noise model as described in the main text, where Pauli errors occur independently between qubits and bit-flips independently on each syndrome bit.
 The error rates can be different on each data qubit and syndrome bit. 
 The maximum-likelihood decoder, described in the main text, can be easily modified to include these measurement errors. This is done simply by including the measurement errors as additional nodes, connected to the factor corresponding to the syndrome bit that they flip.
 This does not destroy the tree structure, and thus decoding and determination of marginals can still be done via \ac{BP}.
 Using this adapted maximum-likelihood decoder, we can estimate error rates in the same way as described in the main text.
 It should be noted that we do not consider a fault-tolerant scheme with repeated measurements here, so identification of measurement errors is impossible on the first concatenation level.
 Similar to the experiments in the main text, we take the data qubits to be affected by a depolarizing channel with error rate $p$ each, and on each syndrome bit the outcome is flipped with probability $p_{m}$.
 In the \autoref{fig:ErrorRates_WithMeasurementErrors}, some results are shown.
\begin{figure}[tb]
\subfloat[$\alpha = 20$ \label{fig:ErrorRates_WithMeasurementErrors_p0005_A20}]{
\includegraphics[scale=1]{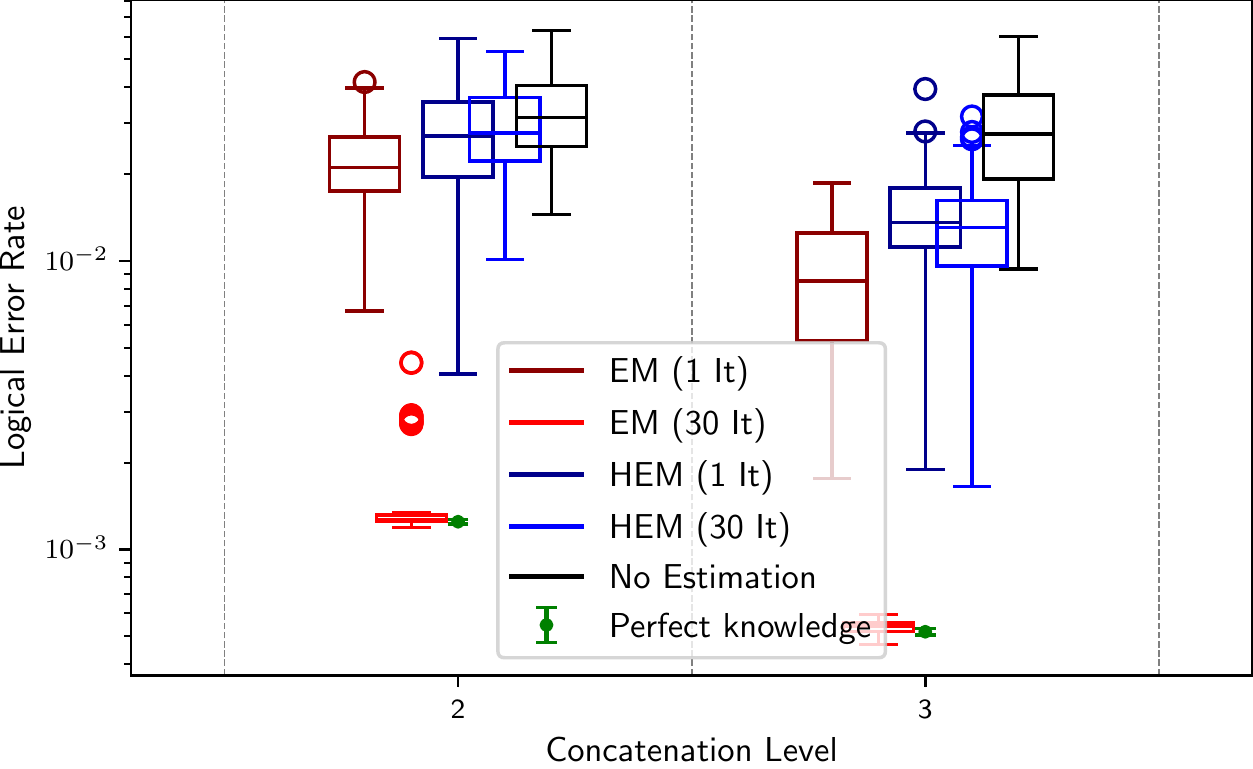}}\\
\subfloat[$\alpha = 200$ \label{fig:ErrorRates_WithMeasurementErrors_p0005_A200}]
{\includegraphics[scale=1]{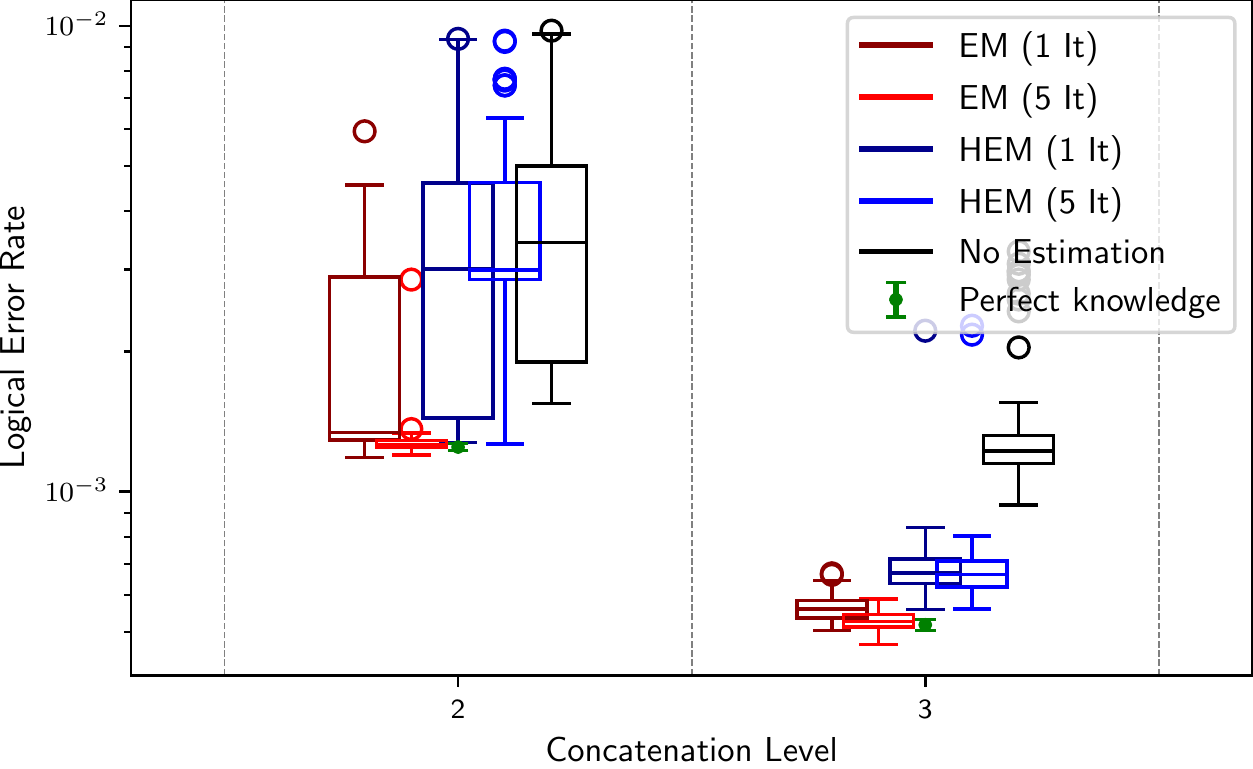}}
\caption{Logical error rate of the maximum likelihood decoder with measurement errors for different values of $\alpha$. Also shown are error rates of the perfect knowledge decoder. The parameters were $p  = 0.005$, $p_m = 0.005$, $\nest = 10^4$. \label{fig:ErrorRates_WithMeasurementErrors}}
\end{figure}
As can be seen in \autoref{fig:ErrorRates_WithMeasurementErrors_p0005_A20}, for a bad initialization \ac{HEM} is unable to improve much over the initialization, while \ac{EM} still reaches optimal error rates even in the presence of measurement errors, although the amount of iterations required is larger than in the case without measurement errors.
The \ac{MSE} of the estimation was again close to the \ac{CRB} (not shown here). The case of a better initialization is shown in \autoref{fig:ErrorRates_WithMeasurementErrors_p0005_A200}.
In this setting, \ac{HEM} clearly improves over the initialization, especially at higher concatenation levels. It is still outperformed by \ac{EM}, and the difference is more significant at the second concatenation level. The amount of iterations before convergence of \ac{EM} is only about 5, compared to about 30 for the bad initialization case.

\FloatBarrier
\section{Conclusion}
We investigated the estimation of stochastic error models from the syndrome statistics of a quantum error correction code, establishing both theoretical results on parameter identifiability as well as a practical estimation method. The results do not rely on the limit of low error rates, and our estimator outperforms other recently proposed methods \cite{Wootton_QiskitBenchmarking,Fowler_ScalableExtractionOfErrorModelsFromQEC,Huo_2017}.  
Our work opens up a number of new research directions.
On the theoretical side, it will be interesting to use our identifiability condition to prove results beyond perfect codes. It might also be possible to extend the result on perfect codes beyond the case of equal rates, since numerical results suggest that this assumption is not crucial.
The proposed estimator could be straightforwardly applied to quantum low density parity check codes, although the problem arises that \acl{BP} is no longer exact in this scenario.
One could also combine our estimator with methods from Refs.~\cite{Cappe_OnlineEM, Huo_2017} to estimate time-dependent error rates and avoid the re-decoding overhead, or consider its application to fault-tolerant circuits as was done for the hard assignment method in Ref.~\cite{Fowler_ScalableExtractionOfErrorModelsFromQEC}.

\section{Code Availability}
Our Python implementation of the estimator is available on \href{https://github.com/TWagner2/NoiseEstimationFromSyndromes}{GitHub}.

\begin{acknowledgments}
This work was funded by the Deutsche Forschungsgemeinschaft (DFG, German Research Foundation) under Germany's Excellence Strategy – Cluster of Excellence Matter and Light for Quantum Computing (ML4Q) EXC 2004/1 – 390534769.
Plots were created using the Matplotlib \cite{Hunter_Matplotlib} library. The simulations made use of the Numpy \cite{numpy} and Numba \cite{Lam_Numba} Python packages.
\end{acknowledgments}

\bibliography{bibliography}

\end{document}